\documentclass[11pt]{amsart}

\usepackage{amsmath}
\usepackage{amssymb,amssymb,framed}
\usepackage{graphicx}
\usepackage{color,xcolor}

\newcommand{\rem}[1]{}
\newcommand{\bom}{\mbox{\boldmath$\omega$}}
\newcommand{\bphi}{\mbox{\boldmath$\phi$}}

\newcommand{\bk}{\mbox{\boldmath$k$}}
\newcommand{\bu}{\mbox{\boldmath$u$}}
\newcommand{\bv}{\mbox{\boldmath$v$}}
\newcommand{\bw}{\mbox{\boldmath$w$}}

\newcommand{\bz}{\mbox{\boldmath$z$}}
\newcommand{\bx}{\mbox{\boldmath$x$}}
\newcommand{\bdf}{\mbox{\boldmath$f$}}

\newcommand{\Rb}{\beta}

\newcommand{\Att}{\mathcal{A}}

\newcommand{\Gamzero}{\Gamma_{0}}
\newcommand{\Gamtwo}{\Gamma_{2}}

\newcommand{\bel}{\begin{equation}\label}
\newcommand{\ee}{\end{equation}}
\newcommand{\beq}{\begin{eqnarray}\label} 
\newcommand{\eeq}{\end{eqnarray}}  
\newcommand{\bea}{\begin{align}} 
\newcommand{\eea}{\end{align}}  
\newcommand{\bc}{\begin{center}} 
\newcommand{\ec}{\end{center}} 
\newcommand{\ben}{\begin{enumerate}}
\newcommand{\een}{\end{enumerate}}
\newcommand{\bit}{\begin{itemize}}
\newcommand{\eit}{\end{itemize}}

\newcommand\shalf{\ensuremath{{\scriptstyle\frac{1}{2}}}}

\usepackage{mathtools}

\newtheorem{proposition}{Proposition}[]
\newtheorem{theorem}[proposition]{Theorem}
\newtheorem{corollary}[proposition]{Corollary}
\newtheorem{lemma}{Lemma}
\theoremstyle{definition}

\newtheorem{remark}[proposition]{Remark}

\numberwithin{equation}{section}

\usepackage[numbers,sort&compress]{natbib}

\rem{
\setcounter{section}{0}
\setcounter{subsection}{0}
\setcounter{theorem}{1}
\setcounter{lemma}{1}
\setcounter{proposition}{0}
\setcounter{page}{1}

\makeatletter
\@addtoreset{equation}{section}

\makeatother
}

\allowdisplaybreaks

\numberwithin{equation}{section}

\title[\small Global attractor of the TTSH equations]{The global attractor of the Toner-Tu-Swift-Hohenberg equations of active turbulence and its properties}

\par\vspace{-2mm}
\author[D. W. Boutros]{Daniel W. Boutros}
\address[D. W. Boutros]{Department of Applied Mathematics and Theoretical Physics, University of Cambridge, Cambridge CB3 0WA UK}
\email{\tt dwb42@cam.ac.uk}

\author[K. V. Kiran]{Kolluru Venkata Kiran}
\address[K. V. Kiran]{Université Côte d’Azur, CNRS, Institut de Physique de Nice, 06200 Nice, France}
\email{{\tt kiran8147@gmail.com}}

\author[J. D. Gibbon]{John D. Gibbon}
\address[J. D. Gibbon]{Department of Mathematics, Imperial College London, London SW7 2AZ, UK}
\email{\tt j.d.gibbon@ic.ac.uk}

\author[R. Pandit]{Rahul Pandit}
\address[R. Pandit]{Centre for Condensed Matter Theory, Department of Physics, Indian Institute of Science, Bengaluru, 560012, India}
\email{\tt rahul@iisc.ac.in}



\date{\today}
\begin{document}
\par\vspace{-2mm}
\begin{abstract}
The Toner-Tu-Swift-Hohenberg (TTSH) equations are one of the basic equations that are used to model turbulent behaviour in active matter, specifically the swarming of bacteria in suspension. They combine features of the incompressible Navier-Stokes, the Toner-Tu and Swift-Hohenberg equations, together with the important properties that they are linearly driven, and that the Laplacian diffusion is taken to be negative in combination with hyper-dissipation. We prove that the TTSH equations possess a finite-dimensional compact global attractor on the periodic domain $\mathbb{T}^d$ ($d=2,3$) and we establish explicit estimates for its Lyapunov dimension which agree with the heuristic prediction based on the Swift-Hohenberg length scale. The predominance of this length scale (as a vortex length scale) has been observed in both numerical and experimental studies of bacterial turbulence, so our methods and results provide a rigorous theoretical foundation for this phenomenon. We also carry out pseudospectral direct numerical simulations of these PDEs in dimension $d=2$ through which we obtain Lyapunov spectra for representative parameter values. We show that our numerical results are consistent with the analytically derived rigorous bounds.
\end{abstract}
\maketitle
\par\vspace{-7mm}
{\centering \textit{\small Dedicated to Professor Peter Constantin, on the occasion of his 75th birthday.} \par}
\par\vspace{-4mm}
\tableofcontents

\section{\large Introduction}

In the subject of active turbulence \cite{Vicsek1995,TT1995,TTR2005,Rama-Rev-2010,AJC2020,ACJ2022,Rama-Rev-2019,MJR2013} the incompressible, deterministic Toner-Tu (ITT) equations are a simplified version of their compressible, stochastically forced counterparts  \cite{CLT2015,RP2020,KGPP2023}. The ITT equations are given by
\bel{ITT}
\left(\partial_{t} + \lambda\bu\cdot\nabla\right)\bu - \Gamma\Delta\bu + \nabla p = -\left(\alpha + \beta|\bu|^{2}\right)\bu\,, \quad \nabla \cdot \bu = 0\,,
\ee
where $\bu : \mathbb{T}^3 \times [0,T] \rightarrow \mathbb{R}^3$ is the velocity field, $p : \mathbb{T}^3 \times [0,T] \rightarrow \mathbb{R}$ is the pressure, $\Gamma > 0$ is the viscosity, and $\alpha \in \mathbb{R}, \beta > 0$ and $\lambda > 0$ are coefficients. In effect, the ITT equations can be viewed as the Navier-Stokes equations with a forcing term $\bdf = -\bu\left(\alpha + \beta|\bu|^{2}\right)$ on the right-hand side. A further significant and important model in active turbulence, and closely related to the ITT equations, are the Toner-Tu-Swift-Hohenberg (TTSH) equations. These were first introduced by Wensink \textit{et al}~\cite{wensink2012} in order to model bacterial suspensions, but they are also relevant to active phenomena where swarming, flocking and shoaling are dominant. Swift and Hohenberg had originally modified a nonlinear diffusion equation, designed to model certain convective processes, by replacing the standard Laplacian viscous term by both anti-diffusive and bi-Laplacian dissipative terms~\cite{SH1}. Incorporating these competing effects into the ITT equations leads to the introduction of the following TTSH equations, which are the main focus of this paper
\bel{TTSHpde}
\left(\partial_{t} + \bu\cdot\nabla\right)\bu + \Gamma_{0}\Delta\bu + \Gamma_{2}\Delta^{2}\bu + \nabla p
= - \left(\alpha + \beta|\bu|^{2}\right)\bu\,, \quad \nabla \cdot \bu = 0\,,
\ee
where $\bu : \mathbb{T}^3 \times [0,T] \rightarrow \mathbb{R}^3$ is the coarse-grained velocity field, $p : \mathbb{T}^3 \times [0,T] \rightarrow \mathbb{R}$ is the pressure, and $\Gamma_{0} > 0$, $\Gamma_{2} > 0$, and $\beta > 0$ are parameters. The parameter $\alpha$ can take either sign\footnote{In the biological literature the linear driving term is usually written as $-\alpha\bu$. Although $\alpha$ can take either sign, experimental studies tend to concentrate on negative values of $\alpha$\,: e.g., in \cite{SSR1} $\alpha$ lies in the range $-8 \leq \alpha \leq -1$. The parameter $\lambda>0$ in the material derivative accounts for the pulling/pushing effect in flocking and shoaling phenomena. For simplicity we take $\lambda=1$, but the results can also be derived for different values of $\lambda$.} and the parameter $\lambda$ in (\ref{ITT}) is taken to be unity. 
\par\smallskip
The TTSH equations~\eqref{TTSHpde} possess certain significant features, such as the existence of steady state solutions $|\overline{\bu}|^{2} = -\alpha/\beta$ (also known as the ordered polar states \cite{wensink2012}), under the assumption that $\alpha < 0$. Moreover, the equations have the property that long-wavelength perturbations around both the zero state and the ordered polar states $\overline{\bu}$ are linearly unstable. In the former case, small perturbations of the type $\varepsilon\exp(i\bk\cdot\bx + \Sigma t)$ solve the linearised TTSH equations around $\bu=0$ if $\Sigma$ satisfies (cf. equation (30) in the appendix of \cite{wensink2012})
\bel{gr1}
\Sigma = -\alpha + \Gamma_{0}\lvert \bk \rvert^{2} - \Gamma_{2} \lvert \bk \rvert^{4}\,.
\ee
Thus, $\Sigma > 0$ for a finite bandwidth of low wavenumbers $\bk$ such that $\lvert \bk \rvert^{2} < \lvert \bk_{c} \rvert^{2}$, whereas $\Sigma < 0$ for wavenumbers $\bk$ such that $\lvert \bk \rvert^{2} > \lvert \bk_{c} \rvert^{2}$. The critical value $\lvert \bk_{c} \rvert^{2}$ is given by (as derived in equation (31) in the Appendix of \cite{wensink2012})
\bel{SHscale}
\lvert \bk_{c} \rvert^{2} = \frac{\Gamma_{0}}{2\Gamma_{2}}\left\{1 + \left(1 - 4\alpha\frac{\Gamma_{2}}{\Gamma_{0}^2}\right)^{1/2}\right\}\,.
\ee
Usually $\Gamma_{0} >> \Gamma_{2}$ and $0 < \Gamma_{2}<<1$, so the term in the square root is close to unity (assuming that $\alpha$ is of order unity), leaving us with $\lvert \bk_{c} \rvert^{2} \approx \Gamma_{0}/\Gamma_{2}$. We observe that $\lvert \bk_{c} \rvert$ then approximates the inverse of what is called the Swift-Hohenberg scale $\eta_{\,\textsc{sh}} \coloneqq \Gamma_{2}/\Gamma_{0}$, where the linear anti-diffusion and bi-Laplacian linear terms dominate the dynamics and the short wavelength scales $\eta < \eta_{\,\textsc{sh}}$ are stable, whereas long wavelength scales $\eta > \eta_{\,\textsc{sh}}$ are unstable. Additionally, a linear stability analysis around the ordered polar state $|\overline{\bu}|^{2} = -\alpha/\beta$ can easily be performed by writing 
\bel{steady}
\bu = \overline{\bu} + \varepsilon \exp(i\bk\cdot\bx + \Sigma_{0} t)\,.
\ee
The details can be found in the Appendix in \cite{wensink2012} and also in \cite{bui2019,bui2023}. It is shown in \cite{wensink2012} that $\Sigma_{0} > 0$ for a range of wavenumbers $\bk$ such that $\lvert \bk \rvert^{2} < \lvert \bk_{c} \rvert^{2}$ where $\lvert \bk_{c} \rvert^{2} = \Gamma_{0}/\Gamma_{2}$. Note that here we have omitted the additional restrictions on $\bk_x$\,: cf. equation (42) in the Appendix in \cite{wensink2012}.
\par\smallskip
Furthermore, in the spatio-temporally chaotic (turbulent) regime, linear instabilities give rise to vortical structures with a characteristic size $\sim\mathcal{O}(\eta_{\,\textsc{sh}})$. As a result, the energy spectrum peaks at $\lvert \bk_c \rvert = \eta_{\,\textsc{sh}}^{-1}$, with a finite inverse energy flux toward modes $\bk$ such that $\lvert \bk \rvert < \lvert \bk_c \rvert$, resembling the classical inverse cascade in $2d$ fluid turbulence \cite{bratanov2015new,cp2020friction,Boffetta2012two,pandit2017overview}. In this sense, $\eta_{\,\textsc{sh}}$ acts effectively as a forcing scale. However, this interpretation holds only for moderate values of $\alpha \sim \mathcal{O}(1)$. When $\alpha<0$, the $\alpha\bu$ term also injects energy, and is dominant at long wavelengths. For sufficiently large $\alpha$, vortical structures of size $\mathcal{O}(10\eta_{\,\textsc{sh}})$ appear, shifting the spectral peak to the range of modes $\bk$ with $\lvert \bk \rvert < \lvert \bk_c \rvert$~\cite{mukherjee2021anomalous,SSR1,Pandit_2025,kashyap2025emergence}. Although the dependence of $\bk_c$ on $\alpha$ is already apparent from Eq.~\eqref{SHscale}, nonlinear effects amplify this behaviour significantly. In particular, when $-\alpha\sim \Gamma_0/\Gamma_2$, the critical value $\bk_c$ depends on $\alpha$ at leading order. Numerical simulations have revealed that even for $\alpha\leq-5$, large-scale velocity correlations emerge with characteristic size $\sim \mathcal{O}(10\eta_{\,\textsc{sh}})$~\cite{mukherjee2021anomalous,SSR1,Boff2023,kiran2025onset,kashyap2025emergence}.
\par\smallskip
Given the instability of long wavelength perturbations (i.e. $\lvert \bk \rvert^{2} < \lvert \bk_{c} \rvert^{2}$), the outstanding question is the long-time behaviour of solutions of the TTSH equations.  In the mathematical literature, the global well-posedness of the 3D TTSH equations (for $H^2$-data) has been established on the whole space in \cite{Zanger2016}. In this paper we show that the TTSH equations on the unit periodic domain $\mathbb{T}^d$ ($d=2,3$) possess a finite-dimensional global attractor $\Att$. The unit domain has been chosen for simplicity. In descriptive terms, $\Att$ is a finite-dimensional compact set into which solutions from all initial conditions are attracted. Its properties are therefore key to our understanding of the (asymptotic) dynamics of solutions. Moreover, the existence of $\Att$ leads to the fundamental idea of the system having a finite number of degrees of freedom, as $\Att$ is finite-dimensional.  Following Landau and Lifshitz \cite{LandauLifshitz6}, if $\eta_{res}$ is the smallest resolution scale of the dynamics in a $d$-dimensional domain of volume $L^{d}$, then the number of degrees of freedom $\mathcal{N}$ is defined as
\bel{Ndef}
\mathcal{N} = \left(\frac{L}{\eta_{res}}\right)^{d}\,.
\ee
For a PDE, the Lyapunov dimension $d_{L}(\mathcal{A})$ of the global attractor $\mathcal{A}$ is usually interpreted as $\mathcal{N}$ \cite{CFMT}, an idea which is explained in \S\ref{sect:3}. In this paper, we will show that in the case $\Gamtwo^{-1},\Gamzero \gg \alpha, \beta$, the dimension estimate of the global attractor leads to the following length scale
\begin{equation}
\eta_{res} \sim \eta_{\,\textsc{sh}}\,.
\end{equation}
To be more precise, in the two-dimensional case we will show the following dimension estimate (for some constant $c > 0$)
\begin{equation*}
d_L (\Att) \leq c \max \left\{\left(\frac{\Gamzero^2}{\Gamtwo^2} - \alpha \Gamtwo^{-1}\right)^{1/2}\,;
~~\left\lvert \alpha - \frac{1}{2} \frac{\Gamzero^2}{\Gamtwo} \right\rvert^{3/8}\Gamtwo^{-5/8}\Rb^{-1/4}
\right\}\,,
\end{equation*}
while in the three-dimensional case we will prove the following bound
\begin{equation*} 
d_L (\Att) \leq c \max \left\{\left(\frac{\Gamzero^2}{\Gamtwo^{2}} - \alpha \Gamtwo^{-1} \right)^{3/4}\,;~~
2^{9/4}\Gamtwo^{-1} \Rb^{-1/2} \left\lvert \alpha - \frac{1}{2} \frac{\Gamzero^2}{\Gamtwo} \right\rvert^{1/2}\right\}\,.
\end{equation*}
Therefore, in the relevant parameter ranges (i.e. $\Gamtwo^{-1},\Gamzero \gg \alpha, \beta$) we confirm the predominance of the Swift-Hohenberg scale through the $2d$ and $3d$ attractor dimension estimates that are established in \S\ref{2datt} and \S\ref{3datt} (cf. Remarks \ref{2Dremark} and \ref{3Dremark}). This occurs because the contribution from the nonlinear terms to the asymptotic dynamics is negligible compared to the linear terms in this setting. 
\par\smallskip
In a survey of the literature concerning the TTSH equations we note that they share similar features with the Brinkman-Forchheimer equations, which can be coupled to a temperature or a magnetic field\,: see \cite{markowich,titi2019,titi2022} and references therein. Additionally, there are structural similarities with the Ginzburg-Landau equation \cite{doering1988,ghidaglia1987,BCDGG1989,duan1993,LO1997}. The anti-diffusion and bi-Laplacian terms are features that also appear in the Kuramoto-Sivashinsky (KS) equation which models the diffusive thermal instabilities in a laminar flame front. An extensive list of references can be found in the paper by Kalogirou, Keaveny and Papageorgiou \cite{Papa2015} on the wide range of studies conducted on the KS equations in the last 40 years\,: see also \cite{FNST1988,KNS1990,jolly1990}. The related ITT equations (\ref{ITT}) have also been studied recently in \cite{KGPP2023,bae,choi2024}. In \cite{KGPP2023} the global well-posedness of the $2d$ ITT equations was established, as well as the global existence of weak solutions in the $3d$ case. In \cite{bae} the global existence of strong solutions for the $3d$ ITT equations was established, under the assumption of small initial data. Asymptotic stability and decay rates for several steady states for both the $2d$ and $3d$ cases are also obtained in \cite{bae}. The compressible Toner-Tu equations were studied in \cite{choi2024}, in which the global existence of strong solutions for small initial data and the asymptotic stability of a class of steady states was proved.
\par\smallskip
In a volume celebrating Peter Constantin's birthday, it should be noted that many of his pioneering techniques and ideas will be used in later sections, such as\,: i) how to estimate the Lyapunov dimension of a global attractor for a dissipative PDE using the idea of global Lyapunov exponents \cite{ConstantinFoiasattractors,CFT2,ConstantinFoias,CFMT,Constantin1987}\,; and ii) the introduction and use in this setting of the Lieb-Thirring inequalities \cite{Lieb1,Ruelletrace,CFT2}. In \S\ref{sect:2} we will establish the global existence and uniqueness of weak solutions for the TTSH equations and the existence of a global attractor $\mathcal{A}$, while in \S\ref{sect:3} we will obtain analytical estimates for $d_{L}(\mathcal{A})$ in both two and three dimensions. 
\par\smallskip
In addition to these results, in \S\ref{Kiran} we also perform pseudospectal numerical calculations of Lyapunov spectra and, by extension, compute the attractor dimension using the vorticity formulation as opposed to the velocity formulation, a choice which reduces the computational cost. While our analytical estimates have been performed by using the velocity formulation of the TTSH equations, in \S\ref{sect:2} we demonstrate that $\Att_{L^2} = \Att_{H^1}$ (using standard techniques), thereby ensuring that the attractor is identical in the two formulations of the equations. In \S\ref{conclusions} we will make some concluding remarks and discuss the significance of our results.

\section{\large Proof of global well-posedness and the existence of a global attractor}\label{sect:2}

First we will demonstrate that the TTSH equations (\ref{TTSHpde}) are globally well-posed, both in dimensions two and three. The global well-posedness of these equations (on the whole space) has already been established in \cite{Zanger2016}. However, to deduce the existence of a global attractor in the case of a periodic domain we will need different a priori estimates, in particular, to establish the existence of an absorbing ball. Thus, we give a new proof of the global well-posedness of the TTSH equations below. Throughout \S \ref{sect:2} and \S \ref{sect:3} we will restrict to the case of the domain having unit length, we will restore the units of length later in \S \ref{Kiran}.
\begin{theorem} \label{weaksolutionsthm}
Let $\bu_0 \in L^2 (\mathbb{T}^d)$, $T > 0$ and $d=2,3$. Then there exists a unique global weak solution $\bu \in C([0,T];L^2(\mathbb{T}^d)) \cap L^2 ((0,T); H^2 (\mathbb{T}^d)) \cap L^4 ((0,T); L^4 (\mathbb{T}^d))$ of the TTSH equations.
\end{theorem}
\begin{proof}
We will prove the existence of solutions by means of the Galerkin method. Let $P_N$ be the projection onto the first $N$ eigenfunctions of the Stokes operator, where the zeroth mode is also included.  We will consider the (truncated) Galerkin system with the divergence-free constraint $\mbox{div}\,\bu_N = 0$
\begin{align}
&\partial_{t} \bu_N + P_N [\bu_N \cdot\nabla \bu_N] + \nabla p_N \nonumber \\
&= -\alpha\bu_N - \Rb\,P_N [\bu_N|\bu_N|^{2}] - \Gamzero \Delta\bu_N - \Gamtwo \Delta^{2}\bu_N\,. \label{galerkinsystem}
\end{align}
The Galerkin system \eqref{galerkinsystem} is a coupled system of ODEs with quadratic nonlinearities, so we deduce that there exists a unique local-in-time solution to equation \eqref{galerkinsystem}. We take the $L^2 (\mathbb{T}^d)$ inner product of equation \eqref{galerkinsystem} with $\bu_N$, which gives
\begin{align*}
&\frac{1}{2} \frac{d}{d t} \lVert \bu_N \rVert_{L^2}^2 + \alpha \lVert \bu_N \rVert_{L^2}^2 + \Rb \lVert \bu_N \rVert_{L^4}^4 - \Gamzero \lVert \nabla \bu_N \rVert_{L^2}^2 + \Gamtwo \lVert \Delta \bu_N \rVert_{L^2}^2 = 0\,.
\end{align*}
By using an interpolation estimate and Young's inequality, one observes 
\begin{align}
\Gamzero \lVert \nabla \bu_N \rVert_{L^2}^2 &\leq \Gamzero \lVert \bu_N \rVert_{L^2} \lVert \Delta \bu_N \rVert_{L^2} \nonumber \\
&\leq \frac{1}{2} \Gamzero^{2}\,\Gamtwo^{-1} \lVert \bu_N \rVert_{L^2}^2 + \frac{1}{2} \Gamtwo \lVert \Delta \bu_N \rVert_{L^2}^2\,. \label{antidiffusionyoungestimate}
\end{align}
By using H\"{o}lder's inequality we find
\begin{equation*}
\lVert \bu_N \rVert_{L^2} \leq \lVert \bu_N \rVert_{L^4}\,,
\end{equation*}
on a unit periodic domain $[0,\,1]^d$. For the sake of simplicity, we will now introduce the parameter
\begin{equation}
\alpha_{R} \coloneqq \alpha - \frac{1}{2} \frac{\Gamzero^2}{\Gamtwo}\,.
\end{equation}
Therefore, we can rewrite the $L^2$-estimate as follows 
\begin{align*}
\frac{1}{2} \frac{d}{d t} \lVert \bu_N \rVert_{L^2}^2 + \alpha_{R}\lVert \bu_N \rVert_{L^2}^2 + 
\frac{\Rb}{2} \lVert \bu_N \rVert_{L^2}^{4} + \frac{\Rb}{2} \lVert \bu_N \rVert_{L^4}^{4} + \frac{1}{2} \Gamtwo \lVert \Delta \bu_N \rVert_{L^2}^2 \leq 0\,.
\end{align*}
We will restrict ourselves to the case $\alpha_{R} < 0$; otherwise all solutions will converge to zero.\\
To estimate the contribution of the $\alpha_{R}\lVert \bu_N \rVert_{L^2}^2$-term we again use Young's inequality to find
\begin{equation}
\alpha_{R}\lVert \bu_N \rVert_{L^2}^2 \geq - \frac{\Rb}{2}\lVert \bu_N \rVert_{L^2}^{4} - 2 \frac{\alpha_R^2}{\Rb} + \lvert \alpha_{R}\rvert \lVert \bu_N \rVert_{L^2}^{2}\,.
\end{equation}
Inserting this relation into the previous estimate then gives
\begin{align}\label{L2estimate}
&\frac{1}{2} \frac{d}{d t} \lVert \bu_N \rVert_{L^2}^2 + \lvert \alpha_{R}\rvert \lVert \bu_N \rVert_{L^2}^{2} + \frac{\Rb}{2} \lVert \bu_N \rVert_{L^4}^4 + \frac{1}{2} \Gamtwo \lVert \Delta \bu_N \rVert_{L^2}^2 \leq 2 \frac{\alpha_{R}^2}{\Rb}\,.
\end{align}
From this estimate we can deduce that the Galerkin approximations are global-in-time. After dropping the hyper-diffusion term from \eqref{L2estimate} as well as the contribution from the cubic term, we find by applying Gronwall's inequality
\begin{equation}
\lVert \bu_N (\cdot, t) \rVert_{L^2}^2 \leq \lVert \bu_0 \rVert_{L^2}^2 e^{-2|\alpha_{R}|t} +
2 \frac{|\alpha_{R}|}{\Rb}(1 - e^{-2|\alpha_{R}|t})\,.
\end{equation}
Hence, for any $\bu_0 \in L^2 (\mathbb{T}^d)$, there exists a time $T^*$ (depending on $\lVert \bu_0 \rVert_{L^2}$) such that
\begin{equation} \label{absorbingballL2}
\lVert \bu_N (\cdot, t) \rVert_{L^2}^2 \leq 4 \frac{|\alpha_{R}|}{\Rb} =: k_{1}\,,
\end{equation}
where the bound is independent of $N$. Now we integrate equation \eqref{L2estimate} in time, which leads to (again for $t \geq T^*$)
\begin{align} 
\frac{\Rb}{2} \int_t^{t+1} \lVert \bu_N \rVert_{L^4}^4 dt' &+ \frac{1}{2} \Gamtwo \int_t^{t+1} \lVert \Delta \bu_N \rVert_{L^2}^2 dt' \nonumber \\
&\leq \shalf \lVert \bu_N (\cdot, t) \rVert_{L^2}^2 + 2 \frac{\lvert \alpha_{R} \rvert^2}{\Rb} \leq \shalf k_1 + 2
\frac{\lvert \alpha_{R} \rvert^2}{\Rb} \leq 4
\frac{\lvert \alpha_{R} \rvert^2}{\Rb} \eqqcolon k_{2}\,, \label{L2absorbingball1}
\end{align}
where we have made the assumption (for the sake of simplicity) that $\lvert \alpha_R \rvert \geq 1$. Moreover, in the above computation it is understood that the factor of unity in the integral limit $t+1$ and on the right hand side of (\ref{L2absorbingball1}) is a unit of time. In the case $\lvert \alpha_R \rvert < 1$, we obtain instead
\begin{align} 
\frac{\Rb}{2} \int_t^{t+1} \lVert \bu_N \rVert_{L^4}^4 dt' &+ \frac{1}{2} \Gamtwo \int_t^{t+1} \lVert \Delta \bu_N \rVert_{L^2}^2 dt' \nonumber \\
&\leq \shalf \lVert \bu_N (\cdot, t) \rVert_{L^2}^2 + 2 \frac{\lvert \alpha_{R} \rvert^2}{\Rb} \leq 2
\frac{\lvert \alpha_{R} \rvert}{\Rb} + 2
\frac{\lvert \alpha_{R} \rvert^2}{\Rb} \leq 3
\frac{\lvert \alpha_{R} \rvert^2}{\Rb} + \frac{1}{\Rb}\,, 
\end{align}
where the additional term $\frac{1}{\Rb}$ can then be included in the attractor dimension estimate, leading to a lower order term (with respect to $\Gamzero$ and $\Gamtwo$). Subsequently, by integrating equation \eqref{L2estimate} in time, we find
\begin{equation} \label{L2absorbingball2}
\lvert \alpha_{R} \rvert \int_t^{t+1} \lVert \bu_N \rVert_{L^2}^2 dt' \leq k_2\,.
\end{equation}
Now we turn to the gradient estimates. Taking the $L^2 (\mathbb{T}^d)$ inner product of equation \eqref{galerkinsystem} with $-\Delta \bu_N$ leads to
\begin{align*}
&\frac{1}{2} \frac{d}{d t} \lVert \nabla \bu_N \rVert_{L^2}^2 + \alpha \lVert \nabla \bu_N \rVert_{L^2}^2 - \Gamzero \lVert \Delta \bu_N \rVert_{L^2}^2 + \Gamtwo \lVert \Delta \nabla \bu_N \rVert_{L^2}^2 \\
&= \Rb \int_{\mathbb{T}^d} \lvert \bu_N \rvert^2 \bu_N \cdot \Delta \bu_N dx + \int_{\mathbb{T}^d} (\bu_N \cdot \nabla \bu_N) \cdot \Delta \bu_N dx\,.
\end{align*}
We recall the following identity (see, e.g., Lemma 1 in~\cite{robinson2014})\,:
\begin{equation*}
\int_{\mathbb{T}^d} \lvert \bu_N \rvert^2 \lvert \nabla \bu_N \rvert^2 dx \leq \int_{\mathbb{T}^d} \lvert \bu_N \rvert^2 \bu_N 
\cdot\left(-\Delta \bu_N\right)dx\,.
\end{equation*}
In addition, in the two-dimensional case we have the standard cancellation 
\begin{equation} \label{2dcancellation}
\int_{\mathbb{T}^d} (\bu_N \cdot \nabla \bu_N) \cdot \Delta \bu_N dx = 0\,,
\end{equation}
because $\bu_N$ is divergence-free (cf. Chapter II, Eq. (A.62) in~\cite{foias2001}). In the three-dimensional case we have
\begin{align*}
\int_{\mathbb{T}^d} (\bu_N \cdot \nabla \bu_N) \cdot \Delta \bu_N dx &\leq \lVert \nabla \bu_N \rVert_{L^3}^3 \lesssim \lVert \nabla \bu_N \rVert_{\dot{H}^{1/2}}^3\\ &\lesssim \lVert \bu_N \rVert_{L^2}^{3/2} \lVert \bu_N \rVert_{\dot{H}^3}^{3/2} \\
&\lesssim \frac{\Gamtwo}{6} \lVert \nabla \Delta \bu_N \rVert_{L^2}^2 + \frac{60}{\Gamtwo^3} \lVert \bu_N \rVert_{L^2}^{6}\,.
\end{align*}
From now on, we will restrict our discussion to the three-dimensional case; the two-dimensional case can be treated in a similar fashion (by using the aforementioned cancellation identity \eqref{2dcancellation}). Using an interpolation estimate we have
\begin{align*}
\Gamzero \lVert \Delta \bu_N \rVert_{L^2}^2 &\leq \Gamzero \lVert  \bu_N \rVert_{L^2}^{2/3} \lVert \nabla \Delta \bu_N \rVert_{L^2}^{4/3} \\
&\leq \frac{1}{3} \frac{\Gamzero^3}{\Gamtwo^2} \lVert \bu_N \rVert_{L^2}^2 + \frac{2}{3} \Gamtwo \lVert \nabla \Delta \bu_N \rVert_{L^2}^{2}\,.
\end{align*}
This then leads to the estimate
\begin{align}
\frac{1}{2} \frac{d}{d t} \lVert \nabla \bu_N \rVert_{L^2}^2 &+ \alpha \lVert \nabla \bu_N \rVert_{L^2}^2 + \frac{1}{6} \Gamtwo \lVert \Delta \nabla \bu_N \rVert_{L^2}^2 + \Rb \int_{\mathbb{T}^d} \lvert \bu_N \rvert^2 \lvert \nabla \bu_N \rvert^2 dx \nonumber \\
&\leq \frac{1}{3} \frac{\Gamzero^3}{\Gamtwo^2} \lVert \bu_N \rVert_{L^2}^2 + \frac{60}{\Gamtwo^{3} \lVert \bu_N \rVert_{L^2}^6}\,.\label{H1estimate}
\end{align}
If we drop the coercive terms, and integrate in time\footnote{It is to be understood that the +1 in the interval $[t,\,t+1]$ is a unit of time and likewise in all the intervals $[t+m,\,t+m+1]$.} from $s$ to $t+1$ [where $s \in [t, t+1)$], we find
\begin{align*}
\lVert \nabla \bu_N (\cdot, t + 1) \rVert_{L^2}^2 &\leq \lVert \nabla \bu_N (\cdot, s) \rVert_{L^2}^2 + \frac{1}{3} \frac{\Gamzero^3}{\Gamtwo^2} \int_s^{t+1} \lVert \bu_N \rVert_{L^2}^2 dt' \\
&+ \frac{60}{\Gamtwo^3} \int_s^{t+1} \lVert \bu_N \rVert_{L^2}^6 dt' + \lvert \alpha \rvert \int_s^{t+1} \lVert \nabla \bu_N (\cdot, s) \rVert_{L^2}^2 dt'\,.
\end{align*} 
Because of estimates \eqref{absorbingballL2} and \eqref{L2absorbingball1} we can integrate with respect to $s$ from $t$ to $t+1$ to obtain (for $t \geq T^* + 1$)
\begin{align} 
\lVert \nabla \bu_N (\cdot, t + 1) \rVert_{L^2}^2 &\leq (\lvert \alpha \rvert + 1) \int_t^{t+1} \lVert \nabla \bu_N (\cdot, s) \rVert_{L^2}^2 dt' + \frac{1}{3} \frac{\Gamzero^3}{\Gamtwo^2} \int_t^{t+1} \lVert \bu_N \rVert_{L^2}^2 dt' \nonumber \\
&+ \frac{60}{\Gamtwo^3} \int_t^{t+1} \lVert \bu_N \rVert_{L^2}^6 dt' \leq k_{3}\,. \label{H1absorbingball}
\end{align}
Taking the $L^2 (\mathbb{T}^d)$ inner product of equation \eqref{galerkinsystem} with $\Delta^2 \bu_N$ and integrating by parts twice, leads to
\begin{align*}
&\frac{1}{2} \frac{d}{d t} \lVert \Delta \bu_N \rVert_{L^2}^2 + \alpha \lVert \Delta \bu_N \rVert_{L^2}^2 - \Gamzero \lVert \nabla \Delta \bu_N \rVert_{L^2}^2 + \Gamtwo \lVert \Delta^2 \bu_N \rVert_{L^2}^2 \\
&= - \Rb \int_{\mathbb{T}^d} \lvert \bu_N \rvert^2 \bu_N \cdot \Delta^2 \bu_N dx - \int_{\mathbb{T}^d} (\bu_N \cdot \nabla \bu_N) \cdot \Delta^2 \bu_N dx\,.
\end{align*}
By applying Young's inequality we find that
\begin{align*}
&\frac{1}{2} \frac{d}{d t} \lVert \Delta \bu_N \rVert_{L^2}^2 + \alpha_{R} \lVert \Delta \bu_N \rVert_{L^2}^2 + \frac{\Gamtwo}{2} \lVert \Delta^2 \bu_N \rVert_{L^2}^2 \\
&= - \Rb \int_{\mathbb{T}^d} \lvert \bu_N \rvert^2 \bu_N \cdot \Delta^2 \bu_N dx - \int_{\mathbb{T}^d} (\bu_N \cdot \nabla \bu_N) \cdot \Delta^2 \bu_N dx \\
&\leq \frac{\Gamtwo}{4} \lVert \Delta^2 \bu_N \rVert_{L^2}^2 + \frac{2 \Rb^2}{\Gamtwo} \lVert \bu_N \rVert_{L^6}^6 + \frac{2}{\Gamtwo} \lVert \bu_N \rVert_{L^6}^2 \lVert \nabla \bu_N \rVert_{L^3}^{2}\,,
\end{align*}
for the three-dimensional case, whence we deduce that
\begin{align*}
&\frac{1}{2} \frac{d}{d t} \lVert \Delta \bu_N \rVert_{L^2}^2 + \alpha_{R} \lVert \Delta \bu_N \rVert_{L^2}^2 + \frac{\Gamtwo}{4} \lVert \Delta^2 \bu_N \rVert_{L^2}^2 \\
&\leq \frac{2 \Rb^2}{\Gamtwo} \lVert \nabla \bu_N \rVert_{L^2}^6 + \frac{2}{\Gamtwo} \lVert \bu_N \rVert_{H^1}^{3} \lVert \Delta \bu_N \rVert_{L^2}\,.
\end{align*}
Then, by integrating in time from $s$ to $t+2$, where $s \in [t+1, t+2)$, we find
\begin{align*}
\lVert \Delta \bu_N (\cdot, t+2) \rVert_{L^2}^2 &\leq \lVert \Delta \bu_N (\cdot, s) \rVert_{L^2}^2 + \int_s^{t+2} \bigg[ \frac{2 (\Rb^2 + 1)}{\Gamtwo} \lVert \bu_N \rVert_{H^1}^{6} \\
&+ \left( \frac{1}{4 \Gamtwo} + \lvert \alpha_{R} \rvert \right) \lVert \Delta \bu_N \rVert_{L^2}^2 \bigg] dt'.
\end{align*}
Next, by integrating the variable $s$ from $t+1$ to $t+2$, we find that
\begin{align*}
\lVert \Delta \bu_N (\cdot, t+2) \rVert_{L^2}^2 &\leq \int_{t+1}^{t+2} \lVert \Delta \bu_N (\cdot, s) \rVert_{L^2}^2 ds + \int_{t+1}^{t+2} \bigg[ \frac{2 (\Rb^2 + 1)}{\Gamtwo}\lVert \bu_N \rVert_{H^1}^{6} \\
&+ \left( \frac{1}{4 \Gamtwo} + \lvert \alpha_{R} \rvert \right) \lVert \Delta \bu_N \rVert_{L^2}^2 \bigg] dt'.
\end{align*}
Finally, using estimates \eqref{absorbingballL2}, \eqref{L2absorbingball1} and \eqref{H1absorbingball}, we deduce that there exists a constant $k_4$ such that, for $t \geq T^* + 2$, we have
\begin{equation} \label{H2absorbingball}
\lVert \Delta \bu_N (\cdot, t+2) \rVert_{L^2}^2 \leq k_4.
\end{equation}
From estimates \eqref{L2estimate} and \eqref{H1estimate} we find that, for any $T>0$,
\begin{equation} \label{aprioribounds}
\lVert \bu_N \rVert_{L^\infty ((0,T); L^2 (\mathbb{T}^d)} + \lVert \bu_N \rVert_{L^4 ((0,T); L^4 (\mathbb{T}^d))} + \lVert \bu_N \rVert_{L^2 ((0,T); H^2 (\mathbb{T}^d))} \leq K,
\end{equation}
for some constant $K > 0$ which is independent of $N$, but its value may change from line to line. Using estimate \eqref{aprioribounds} and equation \eqref{galerkinsystem} we find that
\begin{align*}
\lVert \Delta^2 \bu_N \rVert_{L^2 ((0,T); H^{-2} (\mathbb{T}^d))} &\leq K, \\
\lVert \Delta \bu_N \rVert_{L^2 ((0,T); L^2 (\mathbb{T}^d))} &\leq K, \\
\lVert \nabla \cdot (\bu_N \otimes \bu_N) \rVert_{L^2 ((0,T); H^{-1} (\mathbb{T}^d))} &\leq K, \\
\lVert \lvert \bu_N \rvert^2 \bu_N \rVert_{L^2 ((0,T); H^{-2} (\mathbb{T}^d))} &\leq \lVert \bu_N \rVert_{L^6 ((0,T); L^3 (\mathbb{T}^d))}^3 \\
&\leq \lVert \bu_N \rVert_{L^\infty ((0,T); L^2 (\mathbb{T}^d))} \lVert \bu_N \rVert_{L^4 ((0,T); L^4 (\mathbb{T}^d))}^2 \leq K,
\end{align*}
from which it follows that
\begin{equation}
\lVert \partial_t \bu_N \rVert_{L^2 ((0,T); H^{-2} (\mathbb{T}^d))} \leq K.
\end{equation}
Now by applying the Banach-Alaoglu theorem and the Aubin-Lions lemma we find that (by passing to a subsequence if necessary, which we will also label with $\{ \bu_N \}$)
\begin{align*}
&\bu_N \overset{\ast}{\rightharpoonup} \bu \quad \text{in } L^\infty ((0,T); L^2 (\mathbb{T}^d)), \\
&\bu_N \rightharpoonup \bu \quad \text{in } L^2 ((0,T); H^2 (\mathbb{T}^d)), \\
&\bu_N \rightharpoonup \bu \quad \text{in } L^4 ((0,T); L^4 (\mathbb{T}^d)), \\
&\bu_N \rightarrow \bu \quad \text{in } L^2 ((0,T); L^2 (\mathbb{T}^d)), \\
&\partial_t \bu_N \rightharpoonup \partial_t \bu \quad \text{in } L^2 ((0,T); H^{-2} (\mathbb{T}^d)),
\end{align*}
where $\bu \in L^\infty ((0,T); L^2 (\mathbb{T}^d)) \cap L^4 ((0,T); L^4 (\mathbb{T}^d)) \cap L^2 ((0,T); H^2 (\mathbb{T}^d))$. It is clear that we can pass to the limit in the weak formulation of the TTSH equations, and therefore $\bu$ is a weak solution. Now by using the Lions-Magenes lemma [see Theorem II.5.12 in~\cite{boyer2012}] we deduce that $\bu \in C([0,T];L^2 (\mathbb{T}^d))$ and that $\bu$ attains the initial data $\bu_0$. It remains to show that the weak solution is unique. Suppose that there exist two weak solutions $(\bu_1, p_1)$ and $(\bu_2, p_2)$ with the same initial data $\bu_0$\,; we denote their difference by
\begin{equation*}
v \coloneqq \bu_1 - \bu_2, \quad q \coloneqq p_1 - p_2\,.
\end{equation*}
The difference satisfies
\begin{align*}
\partial_t v &+ u_1 \cdot \nabla v + v \cdot \nabla u_2 + \nabla q + \alpha v + \Rb ( \bu_1 \lvert \bu_1 \rvert^2 - \bu_2 \lvert \bu_2 \rvert^2 ) + \Gamzero \Delta v \\
&+ \Gamtwo \Delta^2 v = 0\,.
\end{align*}
First we derive the following inequality
\begin{align*}
&-\int_{\mathbb{T}^d} ( \bu_1 \lvert \bu_1 \rvert^2 - \bu_2 \lvert \bu_2 \rvert^2 ) \cdot v dx = - \int_{\mathbb{T}^d} ( v \lvert \bu_1 \rvert^2 + \bu_2 (v \cdot (\bu_1 + \bu_2)) ) \cdot v dx \\
&\leq (\lVert \bu_1 \rVert_{L^\infty}^2 + \lVert \bu_2 \rVert_{L^\infty}^2) \lVert v \rVert_{L^2}^2\,.
\end{align*}
Therefore the difference satisfies (where again we focus on the three-dimensional case)
\begin{align*}
&\frac{1}{2} \frac{d}{d t} \lVert v \rVert_{L^2}^2 + \alpha \lVert v \rVert_{L^2}^{2} - \Gamzero \lVert \nabla v \rVert_{L^2}^2 + \Gamtwo \lVert \Delta v \rVert_{L^2}^2 \\
&\leq \lVert v \rVert_{L^3} \lVert v \rVert_{L^6} \lVert \nabla u_2 \rVert_{L^2} + (\lVert \bu_1 \rVert_{L^\infty}^2 + \lVert \bu_2 \rVert_{L^\infty}^2) \lVert v \rVert_{L^2}^{2} \\
&\leq \lVert v \rVert_{L^2}^{5/4} \lVert \Delta v \rVert_{L^2}^{3/4} \lVert \nabla u_2 \rVert_{L^2} + (\lVert \bu_1 \rVert_{L^\infty}^2 + \lVert \bu_2 \rVert_{L^\infty}^2) \lVert v \rVert_{L^2}^2 \\
&\leq \frac{1}{\Gamtwo^{3/5}} \lVert \nabla u_2 \rVert_{L^2}^{8/5} \lVert v \rVert_{L^2}^2 + \frac{1}{4} \Gamtwo \lVert \Delta v \rVert_{L^2}^2 + (\lVert \bu_1 \rVert_{L^\infty}^2 + \lVert \bu_2 \rVert_{L^\infty}^2) \lVert v \rVert_{L^2}^{2} \,,
\end{align*}
where we have used the Lions-Magenes lemma, as $\partial_t v \in L^2 ((0,T); H^{-2} (\mathbb{T}^d))$).
Again we use Young's inequality to estimate the contribution from the anti-diffusion term (similarly to estimate \eqref{antidiffusionyoungestimate}), which gives
\begin{align} \label{differenceestimate}
\frac{1}{2} \frac{d}{d t} \lVert v \rVert_{L^2}^2 &+ \alpha_{R} \lVert v \rVert_{L^2}^2 + \Rb \lVert v \rVert_{L^4 (\mathbb{T}^d)}^4 + \frac{1}{4} \Gamtwo \lVert \Delta v \rVert_{L^2}^2 \\
&\leq \Gamtwo^{-3/5} \lVert \nabla u_2 \rVert_{L^2}^{8/5} \lVert v \rVert_{L^2}^2\,. \nonumber
\end{align}
Applying the Gronwall inequality leads to the uniqueness of weak solutions, as well as the continuous dependence on the initial data.
\end{proof}
The a priori estimates in the proof of Theorem \ref{weaksolutionsthm} also lead to the global existence of strong solutions for the TTSH equations \eqref{TTSHpde}.
\begin{corollary}
Let $\bu_0 \in H^1 (\mathbb{T}^d)$ and $T > 0$, then if $d=2,3$ there exists a unique (global) strong solution $\bu \in C([0,T]; H^1 (\mathbb{T}^d)) \cap L^2 ((0,T); H^3 (\mathbb{T}^d))$ to the TTSH equations \eqref{TTSHpde}.
\end{corollary}
\begin{proof}
The regularity of the strong solution can be obtained from estimate \eqref{H1estimate} as well as estimate \eqref{H1absorbingball}, and by again applying the Banach-Alaoglu theorem, as well as the Aubin-Lions lemma. The uniqueness has already been shown in the proof of Theorem \ref{weaksolutionsthm}.
\end{proof}
\begin{corollary} \label{absorbingballresult}
The TTSH equations possess absorbing balls in $L^2 (\mathbb{T}^d)$, $H^1 (\mathbb{T}^d)$ and $H^2 (\mathbb{T}^d)$ if $d = 2,3$. In particular, for any $\bu_0 \in L^2 (\mathbb{T}^d)$, there exists a time $T^* > 0$ such that for any $t \geq T^*$
\begin{align}
\lVert \bu (\cdot, t) \rVert_{L^2} &\leq 2 k_1, \\
\lVert \nabla \bu (\cdot, t) \rVert_{L^2} &\leq 2 k_3, \\
\lVert \Delta \bu (\cdot, t) \rVert_{L^2} &\leq 2 k_4\,,
\end{align}
where $k_1$, $k_3$ and $k_4$ depend on the system parameters (i.e. $\Gamzero,\Gamtwo,\alpha$ and $\Rb$).
\end{corollary}
\begin{proof}
The existence of these absorbing balls follows from observing that the (unique) weak solution $\bu$ constructed in Theorem \ref{weaksolutionsthm} also obeys estimates \eqref{absorbingballL2}, \eqref{H1absorbingball} and \eqref{H2absorbingball}.
\end{proof}
We recall that the global attractor is the largest compact invariant set, which attracts all bounded subsets. Using the existence of the absorbing balls proved in Corollary \ref{absorbingballresult}, in the next theorem we prove that the TTSH equations possess a global attractor in two and three dimensions.
\begin{theorem} \label{globalattractorthm}
The TTSH equations possess a global attractor $\Att_{L^2}$ in $L^2 (\mathbb{T}^d)$ and $\Att_{H^1}$ in $H^1 (\mathbb{T}^d)$ respectively for $d=2,3$, which can be characterized by 
\begin{align}
\Att_{L^2} &= \bigcap_{t > 0} \left(\overline{\bigcup_{t' \geq t} S(t') B_{2(k_1 + k_3), H^1}}^{L^2}\right), \\
\Att_{H^1} &= \bigcap_{t > 0} \left(\overline{\bigcup_{t' \geq t} S(t') B_{2(k_1 + k_3 + k_4), H^2}}^{H^1}\right),
\end{align}
where $B_{2(k_1 + k_3), H^1}$ and $B_{2(k_1 + k_3 + k_4), H^2}$ are balls of radii $2(k_1 + k_3)$ and $2(k_1 + k_3 + k_4)$ centred around $0$ in, respectively, the (strong) $L^2$- or $H^1$-topologies. Moreover, it holds that
\begin{equation}
\Att_{L^2} = \Att_{H^1},
\end{equation}
as subsets of $L^2 (\mathbb{T}^d)$ (and hence also of $H^1 (\mathbb{T}^d)$).
\end{theorem}
\begin{proof}
Let us define the mapping $S : \mathbb{R}_+ \times L^2 (\mathbb{T}^d) \rightarrow L^2 (\mathbb{T}^d)$ by $S_t (\bu_0) \mapsto \bu (t)$; we have shown in the proof of Theorem \ref{weaksolutionsthm} that this mapping defines a semigroup. Moreover, we have shown in Corollary \ref{absorbingballresult} that there exists a compact absorbing ball for the TTSH equations. Therefore, by Theorem 1.1 in \cite[Chapter 1]{temam1997} (or Theorem 10.5 in \cite{robinson2001}), it follows that the TTSH equations possess a compact global attractor in $L^2 (\mathbb{T}^d)$. Similarly, one can show that the equations also possess a global attractor in $H^1 (\mathbb{T}^d)$.
\par\smallskip
Now we proceed to show that $\Att_{L^2} = \Att_{H^1}$. We follow the approach of the proof of Lemma 12.6 in \cite{robinson2001}. From Corollary \ref{absorbingballresult} it follows that $\Att_{L^2}$ is a bounded set in $H^2 (\mathbb{T}^d)$ and hence $\Att_{L^2}$ is a compact set in $H^1 (\mathbb{T}^d)$. Because $\Att_{H^1}$ is the largest compact invariant set in $H^1 (\mathbb{T}^d)$, it follows that $\Att_{L^2} \subset \Att_{H^1}$. The inclusion $\Att_{H^1} \subset \Att_{L^2}$ follows because $\Att_{H^1}$ is a compact subset of $L^2 (\mathbb{T}^d)$ which is invariant under the semigroup.
\end{proof}
\begin{remark}
We have shown explicitly in Theorem \ref{globalattractorthm} that the attractors $\Att_{L^2}$ and $\Att_{H^1}$ are the same because, in the analytical part of this paper (in Section \ref{sect:3}), we will work with the velocity formulation of the TTSH equations, whereas, in the numerical part (in Section \ref{Kiran}), we will work with the vorticity formulation of the equations.
\end{remark}

\section{\large Estimates for the dimension of the global attractor $\Att$}\label{sect:3}

\subsection{Trace formulae and Lieb-Thirring inequalities}\label{LTin}

The concept of Lyapunov exponents combined with the Kaplan-Yorke formula originated as a phase space argument for ODEs \cite{KaplanYorke}. These exponents control the exponential growth or contraction of volume elements in a finite dimensional phase space. The Kaplan-Yorke formula expresses the balance between volume growth and contraction realized on the attractor in phase space and is used to define its Lyapunov dimension.  It has been rigorously applied to global attractors in dissipative PDEs by Constantin, Foias and Temam \cite{ConstantinFoiasattractors,CFT2,ConstantinFoias,DGbook}.  The formula is the following\,: let $\mu_{n}$ be the Lyupaunov exponents, the Lyapunov dimension $d_{L}(\mathcal{A})$ is defined by (cf. Theorem 3.3 and Remark 3.5 in \cite[Chapter 5]{temam1997})
\bel{dldef}
d_{L} = N_{0} -1 + \frac{\mu_{1}+\ldots + \mu_{N_{0}-1}}{-\mu_{N_{0}}}\,,
\ee           
where the number $N= N_{0}$ is chosen such that
\bel{KY2}
\sum_{n=1}^{N_{0}-1}\mu_{n}\geq 0 
\hspace{1cm}\mbox{but}\hspace{1cm}
\sum_{n=1}^{N_{0}}\mu_{n} < 0\,.
\ee
Note that, according to the definition of $N_{0}$, the ratio of exponents in (\ref{dldef}) satisfies
\bel{KY1}
0 \leq \frac{\mu_{1}+\ldots + \mu_{N_{0}-1}}{-\mu_{N_{0}}} < 1\,,
\ee          
so the Kaplan-Yorke formula generally yields a non-integer dimension such that 
\bel{KY3}
N_{0} - 1 \leq d_{L} < N_{0}\,. 
\ee
In simple terms, the value of $N$ that turns the sign of the sum of the Lyapunov exponents of the dynamical system as in (\ref{KY2}) is that value of $N$, namely $N_{0}$, that bounds above $d_{L}$ and also the Hausdorff dimension $d_{H}$ and fractal dimension $d_{F}$ of the attractor \cite{ConstantinFoiasattractors,CFT2,ConstantinFoias,DGbook} (cf. Theorem 3.3 and Remark 3.5 in \cite[Chapter 5]{temam1997}).
\par\smallskip
To use the method for PDEs it is necessary to extend the idea of the Lyapunov exponents to global exponents $\mu_{n}$ of $\mathcal{A}$.  In case of the TTSH equations the phase space is given by $W \coloneqq \{ \bu \in L^2 (\mathbb{T}^d) \, \lvert \, \nabla \cdot \bu = 0 \}$, which is infinite dimensional. The space $W$ possesses an orthonormal basis. In effect, the solution $\bu(t)$ forms an orbit in the space $W$.  Now we take different sets of initial conditions
$\bu(0) + \epsilon \delta\bu_{i}(0)$ (for positive $\epsilon \ll 1$) which evolve at the linear level into $\bu(t)+\epsilon \delta\bu_{i}(t)$, where the $\delta\bu_{i}(0)$ are (distinct) elements of the orthonormal basis of $W$ for $i =1,\ldots,N$. In other words, we have
\bel{trace1}
\partial_t (\delta \bu) =\mathcal{M}\delta \bu, \quad \delta \bu \lvert_{t = 0} = \delta\bu_{i}(0)\,,
\ee
where $\mathcal{M}$ is the linearised operator of the spatial part of the TTSH equations around the solution $\bu (t)$. If they are chosen to be linearly independent, at the initial time these $\delta\bu_{i}$ form an $N$-volume or parallelepiped of volume
\bel{vol1}
V_{N}(t) = \left|\delta\bu_{1}\wedge\delta\bu_{2}\ldots\wedge\delta\bu_{N}\right|
\ee
which changes along the orbit so we need to analyze its time evolution. We define the operator $\mathbf{P}_{N}(t)$ as the $L^{2} (\mathbb{T}^d)$-orthogonal projection onto the finite dimensional subspace $\{\delta\bu_{1},\delta\bu_{2},\ldots,\delta\bu_{N}\}$.  This is given by\footnote{Whether one uses 
$\mathbf{P}_{N}\mathcal{M}\mathbf{P}_{N}$ or $\mathcal{M}\mathbf{P}_{N}$ as in \cite{DGbook} is purely notational.} 
\bel{vol2a}
\dot{V}_{N} = V_{N} Tr \left[\mathbf{P}_{N}\mathcal{M}\mathbf{P}_{N}\right] ,
\ee
which is easily solved to give  
\bel{trace2} 
V_{N}(T) = V_{N}(0) 
\exp \left( \int_{0}^{T}
Tr \left[\mathbf{P}_{N}\mathcal{M}\mathbf{P}_{N}\right](t)dt\right)\,.
\ee
See \cite{robinson2001,DGbook} for more details. Next, we introduce the finite-time average $\left<\cdot\right>_T$, which is defined as follows
\begin{equation}
\left< f \right>_T \coloneqq \frac{1}{T} \int_0^T f(t) dt,
\end{equation}
where $f : [0,T] \rightarrow \mathbb{R}$ is an $L^1 ((0,T))$ function.
The sum of the first $N$ global Lyapunov exponents is then bounded by (see \cite{temam1997} for more details)
\bel{global1}
\sum_{n=1}^{N}\mu_{n} \leq \left<Tr \left[\mathbf{P}_{N}\mathcal{M}\mathbf{P}_{N}\right]\right>_{T}\,, 
\ee
where we are projecting onto the orthonormal set $\{\bphi_{n}\}_{n=1}^N$ (which is a subset of the orthonormal basis for $W$) such that 
\bel{PNdef}
Tr\left[\mathbf{P}_{N}\mathcal{M}\mathbf{P}_{N}\right]\ = \sum_{n=1}^{N} \int_{\mathbb{T}^d} \bphi_{n}\mathcal{M}\bphi_{n}\,dV \,.
\ee
We want to find the value of $N_{0}$ that turns the sign of $\left<Tr\left[\mathbf{P}_{N}\mathcal{M}\mathbf{P}_{N}\right]\right>_T$ such that volume elements contract to zero. This value of $N_{0}$ bounds the Lyapunov dimension $d_{L}(\mathcal{A})$ (and also the Hausdorff and fractal dimensions) of the attractor as defined above in \eqref{dldef}.
\par\smallskip
Because the $\bphi_{n}$ are orthonormal in the space $W$, they obey the relations
\bel{trrelns1}
Tr \left[\mathbf{P}_{N}(-\Delta)\mathbf{P}_{N}\right] 
= \sum_{n=1}^{N} \int_{\mathbb{T}^d} |\nabla\bphi_{n}|^{2}\,dV
\ee
and we also have
\bel{trrelns2}
\sum_{n=1}^{N} \int_{\mathbb{T}^d} |\bphi_{n}|^{2}\,dV = N\,.
\ee
Fortunately, the functions $\bphi_{n}$ have some useful properties of which we can take advantage to estimate the relative magnitudes of the negative and positive contributions to the Lyapunov exponents.  In $d$ spatial dimensions the $\bphi_{n}$ satisfy what are known as Lieb-Thirring inequalities \cite{ConstantinFoias,Lieb1,Ruelletrace} for all orthonormal functions, which we state in the next lemma.
\begin{lemma}
Let $\{\bphi_1, \ldots, \bphi_N \}$ be an orthonormal set in $\{ \bu \in L^2 (\mathbb{T}^d) \, \lvert \, \nabla \cdot \bu = 0 \}$, then the following inequalities (known as the Lieb-Thirring inequalities) hold
\begin{align}\label{LT1A}
\int_{\mathbb{T}^d} \left(\sum_{n=1}^{N}|\bphi_{n}|^{2}\right)^{\frac{d+2}{d}}\,dV
&\leq c\,\sum_{n=1}^{N} \int_{\mathbb{T}^d} |\nabla\bphi_{n}|^{2}\,dV + c\,;\\
\int_{\mathbb{T}^d} \left(\sum_{n=1}^{N}|\bphi_{n}|^{2}\right)^{\frac{d+4}{d}}\,dV
&\leq c\,\sum_{n=1}^{N} \int_{\mathbb{T}^d} |\Delta \bphi_{n}|^{2}\,dV  + c\,,\label{LT1B}
\end{align}
where, $c$ is independent of $N$.  
\end{lemma}\noindent
The proof of these inequalities can be found in \cite[Appendix, Theorem 4.1, Remarks 4.2 and 4.3]{temam1997}.
\par\smallskip
In addition, the Cauchy-Schwarz inequality (both for sums and integrals) ensures that
\bel{LT2}
\left(\sum_{n=1}^{N} \int_{\mathbb{T}^d} |\nabla\bphi_{n}|^{2}\,dV\right)^{2} 
\leq \left(\sum_{n=1}^{N} \int_{\mathbb{T}^d} |\Delta\bphi_{n}|^{2}\,dV\right)
\left(\sum_{n=1}^{N} \int_{\mathbb{T}^d} |\bphi_{n}|^{2}\,dV\right)\,.
\ee
Thus, we have 
\bel{LT3}
Tr\left[\mathbf{P}_{N}(\Delta^2)\mathbf{P}_{N}\right] \geq 
N^{-1}\left(Tr\left[\mathbf{P}_{N}(-\Delta)\mathbf{P}_{N}\right]\right)^{2}\,.
\ee
Moreover, it is well known that for the first $N$ eigenvalues $\lambda_{n}$ of the Stokes operator (in our case with $L=1$)
\bel{LT4}
\sum_{n=1}^{N} \int_{\mathbb{T}^d} |\nabla\bphi_{n}|^{2}\,dV 
= Tr\left[\mathbf{P}_{N}(-\Delta)\mathbf{P}_{N}\right] 
\geq \sum_{j=1}^{N}\lambda_{j} \geq c \sum_{j=1}^{N}j^{2/d}\,.
\ee
Hence we have
\bel{LT5}
Tr\left[{\bf P}_{N}(-\Delta ){\bf P}_{N}\right] \geq c\,N^{\frac{d+2}{d}}\,;
\ee
and it also holds that
\bel{LT6}
Tr\left[\mathbf{P}_{N}(\Delta^2)\mathbf{P}_{N}\right] \geq c\,N^{1+\frac{4}{d}}\,.
\ee
These results can be found in Corollary 4.1 in \cite[Appendix]{temam1997}. Finally, as a preliminary result necessary for proving the estimate on the Lyapunov dimension of the global attractor, we need to show the differentiability of the semigroup. We recall that $\mathcal{M}$ is defined as the linearised operator corresponding to a solution $\bu$, which is the solution map of the following equation (as a map from $L^2 (\mathbb{T}^d)$ to itself):
\begin{align}
\left(\partial_{t} + \bu\cdot\nabla\right)\delta\bu + \delta\bu\cdot\nabla\bu + \nabla \delta p 
&= -\left\{\alpha + \Gamzero \Delta + \Gamtwo \Delta^{2}\right\}\delta\bu\nonumber\\
&- \Rb\,\lvert \bu \rvert^{2}\delta\bu - 2\Rb\,\bu(\delta \bu \cdot \bu )\,. \label{linearisedequation}
\end{align}
\begin{lemma} 
The semigroup $S$ is uniformly differentiable on the attractor $\Att$, i.e. for every $t > 0$ it holds that
\begin{equation*}
\sup_{\bu_0, \bv_0 \in \Att, \; \lVert \bu_0 - \bv_0 \rVert_{L^2} < \epsilon} \frac{\lVert S(t) \bu_0 - S(t) \bv_0 - \mathcal{M} (\bu_0 - \bv_0) \rVert_{L^2}}{\lVert\bu_0 - \bv_0 \rVert_{L^2}} \xrightarrow[]{\epsilon \rightarrow 0} 0\,.
\end{equation*}
To be more precise, we will show that
\begin{equation}
\lVert S(t) \bu_0 - S(t) \bv_0 - \mathcal{M} (\bu_0 - \bv_0) \rVert_{L^2} \leq K \lVert \bu_0 - \bv_0 \rVert_{L^2}^2,
\end{equation}
where the constant $K$ can depend on $t$ and $\mathcal{M}$ is the linearised operator that was introduced previously in equation \eqref{linearisedequation}.
\end{lemma}
\begin{proof}
In this proof we will restrict ourselves to the three-dimensional case (for simplicity); and we will not explicitly track the parameters of the TTSH equations, as we will not use the precise constants later on. In the proof of Theorem \ref{weaksolutionsthm}, we derived that the difference $\bw = S(t) \bu_0 - S(t) \bv_0$ satisfies (and where $q = p_{\small\bu} - p_{\small\bv}$)
\begin{align*}
\partial_t \bw &+ \bu \cdot \nabla \bw + \bw \cdot \nabla \bv + \nabla q + \alpha \bw + \Rb ( \bw \lvert \bu \rvert^2 + \bv ( \bw \cdot (\bu + \bv) ) \\
&+ \Gamzero \Delta \bw + \Gamtwo \Delta^2 \bw = 0\,; \quad \nabla \cdot \bw = 0, \quad \bw \lvert_{t = 0} = \bu_0 - \bv_0.
\end{align*}
First we observe that one can obtain the following identity
\begin{equation} \label{cubicidentity}
\bv ( \bw \cdot (\bu + \bv) ) = 2 \bu (\bw \cdot \bu) - 2 \bw (\bw \cdot \bu) - \bv \lvert \bw \rvert^2\,.
\end{equation}
In the proof of Theorem \ref{weaksolutionsthm}, we showed estimate \eqref{differenceestimate} for $\bw$, from which it follows that there exists a constant $K > 0$ such that
\begin{equation} \label{differencebound}
\lVert \bw (\cdot, t) \rVert_{L^2}^2 \leq K e^{K t} \lVert \bw_0 \rVert_{L^2}^2, \quad \int_0^t \lVert \Delta \bw (\cdot, t') \rVert_{L^2}^2 dt' \leq K e^{K t} \lVert \bw_0 \rVert_{L^2}^2.
\end{equation}
Now we define $\bu_L$ as the solution of the equation
\begin{align}
\big(\partial_{t} &+ \bu\cdot\nabla\big)\bu_L + \bu_L\cdot\nabla\bu + \nabla p_L
= -\left\{\alpha + \Gamzero \Delta + \Gamtwo \Delta^{2}\right\}\bu_L\nonumber\\
&- \Rb\,\lvert \bu \rvert^{2}\bu_L - 2\Rb\,\bu(\bu_L \cdot \bu ), \quad \bu_L (0) = \bu_0 - \bv_0. \nonumber
\end{align}
We also define the function
\begin{equation}
\bz = \bw - \bu_L.
\end{equation}
One can check that $\bz$ satisfies the following equation (by using identity \eqref{cubicidentity})
\begin{align*}
\partial_t \bz &+ \bu \cdot \nabla \bz + \bz \cdot \nabla \bu - \bw \cdot \nabla \bw + \nabla (q - p_L) + \alpha \bz + \Gamzero \Delta \bz + \Gamtwo \Delta^2 \bz \\
&+ \Rb \lvert \bu \rvert^2 \bz + 2 \Rb \bu (\bz \cdot \bu) - 2 \Rb \bw (\bw \cdot \bu) - \Rb \bv \lvert \bw \rvert^2 = 0.
\end{align*}
Now we take the $L^2 (\mathbb{T}^d)$ inner product with $\bz$, which leads to (by applying the Lions-Magenes lemma)
\begin{align*}
&\frac{1}{2} \frac{d}{d t} \lVert \bz \rVert_{L^2}^2 + \alpha \lVert \bz \rVert_{L^2}^2 - \Gamzero \lVert \nabla \bz \rVert_{L^2}^2 + \Gamtwo \lVert \Delta \bz \rVert_{L^2}^2 + \Rb \int_{\mathbb{T}^d} \bigg[ \lvert \bu \rvert^2 \lvert \bz \rvert^2 + 2 \lvert \bz \cdot \bu \rvert^2 \bigg] dx \\
&= \int_{\mathbb{T}^d} \bigg[- \bz \cdot \nabla \bu + \bw \cdot \nabla \bw + 2 \Rb \bw (\bw \cdot \bu) + \Rb \bv \lvert \bw \rvert^2  \bigg] \cdot \bz dx \\
&\leq \lVert \nabla \bu \rVert_{L^2} \lVert \bz \rVert_{L^4}^2 + \lVert \bw \rVert_{L^2} \lVert \nabla \bw \rVert_{L^2} \lVert \bz \rVert_{L^\infty} + 2 \Rb \lVert \bw \rVert_{L^4}^2 (\lVert \bu \rVert_{L^2} + \lVert \bv \rVert_{L^2}) \lVert \bz \rVert_{L^\infty} \,.
\end{align*}
Using Young's inequality repeatedly, together with estimates \eqref{absorbingballL2}, \eqref{H1absorbingball} and \eqref{differencebound}, we find (again for some constant $K > 0$)
\begin{align*}
&\frac{1}{2} \frac{d}{d t} \lVert \bz \rVert_{L^2}^2 + \left(\alpha - \frac{\Gamzero^2}{\Gamtwo} - K \right) \lVert \bz \rVert_{L^2}^2 + \frac{\Gamtwo}{4} \lVert \Delta \bz \rVert_{L^2}^2 \leq K \lVert \bw \rVert_{L^2}^2 \lVert \nabla \bw \rVert_{L^2}^2 + K \lVert \bw \rVert_{L^4}^4.
\end{align*}
Finally, using the Gronwall inequality and $\bz (0) = 0$, we obtain 
\begin{equation}
\lVert \bz (\cdot, t) \rVert_{L^2}^2 \leq K e^{K t} \lVert \bw_0 \rVert_{L^2}^4\,,
\end{equation}
which implies the differentiability of the semigroup.
\end{proof}

\subsection{Attractor dimension estimate in the two-dimensional case}\label{2datt}

A sharp estimate (up to a logarithm) for the dimension of the attractor for the $2d$ NSEs was found by Constantin, Foias and Temam \cite{CFT1988}, aided by Constantin's logarithmic $L^{\infty}$-inequality for functions whose gradients are orthonormal \cite{Constantin1987}. The same dimension estimate in \cite{CFT1988} was later recovered by Doering and Gibbon \cite{DG1991}, who used the vorticity formulation and not the velocity formulation of the Navier-Stokes equations. In the following theorem we use the velocity formulation of the $2d$ TTSH equations as the basis of our calculation, but as was said before by Theorem \ref{globalattractorthm} it makes no difference compared to the vorticity formulation.
\begin{theorem}\label{2dattthm}
The Lyapunov dimension of the global attractor $\mathcal{A}$ of the 2$d$ TTSH equations on the unit periodic domain $\mathbb{T}^2$ is bounded by (for some constant $c$)
\begin{equation*}
d_L (\Att) \leq c\max \left\{\left\lvert\frac{\Gamzero^2}{\Gamtwo^2} - \alpha \Gamtwo^{-1}\right\rvert^{1/2}\,;
~~\left(\frac{|\alpha_{R}|^2}{\Rb}\right)^{3/16}\Gamtwo^{-5/8}\Rb^{-1/16}
\right\}\,,
\end{equation*}
where we recall that $\alpha_{R} = \alpha - \frac{1}{2} \frac{\Gamzero^2}{\Gamtwo}$.
\end{theorem}
\par\vspace{0mm}\noindent
\begin{remark} \label{2Dremark}
By keeping $\alpha$ and $\Rb$ fixed (more specifically if $\Rb^{-1} \ll \Gamzero$), and taking $\Gamtwo \ll 1$ and $\Gamzero \gg 1$ (i.e., the parameter scaling in active turbulence), to leading order we therefore deduce that the dominant term in the attractor dimension estimate is
\begin{equation}
d_L (\Att) \sim \left(\frac{\Gamzero^2}{\Gamtwo^2} - \alpha \Gamtwo^{-1}\right)^{1/2}\,,
\end{equation}
from which we recover the inverse Swift-Hohenberg scale $\eta_{res}^{-1} = (\Gamzero/\Gamtwo)^{1/2}$ by means of relation \eqref{Ndef}. Note that the estimates in the proof of Theorem \ref{2dattthm} also lead to bounds on the Hausdorff and fractal dimensions of the global attractor $\Att$. In addition, the result can also be generalised to include the case of the TTSH equations with a forcing. The numerical methods used in \S\ref{Kiran} require an adjustment of the domain length $L$ which for the above has been taken as unity. A re-working of this calculation on the domain $[0,\,L]^2$ turns this estimate into $\eta_{res}^{-1}\sim L(\Gamzero/\Gamtwo)^{1/2}$.
\end{remark}
\par\smallskip\noindent
\begin{proof}
The linearisation of the 2$D$ TTSH equations around a given solution $\bu$ is given by
\begin{align}
\left(\partial_{t} + \bu\cdot\nabla\right)\delta\bu + \delta\bu\cdot\nabla\bu + \nabla \delta p 
&= -\left\{\alpha + \Gamzero \Delta + \Gamtwo \Delta^{2}\right\}\delta\bu\nonumber\\
&- \Rb\,\lvert \bu \rvert^{2}\delta\bu - 2\Rb\,\bu(\delta \bu \cdot \bu ).
\end{align}
Now let $\{ \bphi_1, \ldots, \bphi_N \}$ be a subset of an orthonormal basis of the space $\{ \bu \in L^2 (\mathbb{T}^d) \, \lvert \, \nabla \cdot \bu = 0 \}$. Then we compute the volume elements as follows:
\begin{align}\label{betaterm}
Tr \left[\mathbf{P}_{N}\mathcal{M}\mathbf{P}_{N}\right]  &= \sum_{n=1}^{N}
\int_{\mathbb{T}^2} \bphi_{n}\cdot\left\{-\alpha\bphi_{n} - \bu\cdot\nabla\bphi_{n}-\bphi_{n}\cdot\nabla\bu 
- \nabla\tilde{p}\left(\bphi_{n}\right)\right\}\,dV\nonumber\\
&-\Rb\sum_{n=1}^{N}\int_{\mathbb{T}^2} \bphi_{n}\cdot\left\{\lvert \bu \rvert^{2}\bphi_{n} + 2 \bu (\bphi_{n} \cdot \bu) \right\}\,dV\,\nonumber\\
&-\sum_{n=1}^{N}\int_{\mathbb{T}^2} \bphi_{n}\cdot\left\{\Gamzero \Delta\bphi_{n} + \Gamtwo \Delta^{2}\bphi_{n}\right\}\,dV\,.
\end{align}
We can bound the trace of the volume element as follows:
\begin{align}
Tr \left[ \mathbf{P}_{N}\mathcal{M}\mathbf{P}_{N}\right] &\leq - \Gamtwo \sum_{n=1}^{N} \int_{\mathbb{T}^2} |\Delta\bphi_{n}|^{2}\!
dV + \Gamzero \sum_{n=1}^{N} \int_{\mathbb{T}^2} |\nabla\bphi_{n}|^{2}\,dV \nonumber \\
&+ \sum_{n=1}^{N} \int_{\mathbb{T}^2}|\bphi_{n}|^{2} \left(-\alpha + |\nabla \bu |\right)\!dV\,. \label{2dtraceestimate}
\end{align}
Next we estimate the terms coming from the advection term:
\begin{align*}
&\sum_{n=1}^{N} \int_{\mathbb{T}^2}|\bphi_{n}|^{2} \lvert \nabla \bu \rvert dV \leq \lVert \nabla \bu \rVert_{L^2} \left\lVert \sum_{n=1}^{N} |\bphi_{n}|^{2} \right\rVert_{L^2} \\
&\lesssim \lVert \nabla \bu \rVert_{L^2} \bigg[ \left( \sum_{n=1}^{N} \int_{\mathbb{T}^2} |\nabla\bphi_{n}|^{2}\,dV \right)^{1/2} + c \bigg] \\
&\lesssim \lVert \nabla \bu \rVert_{L^2} \left(\sum_{n=1}^{N} \int_{\mathbb{T}^2} |\bphi_{n}|^{2}\,dV \right)^{1/4} \left( \sum_{n=1}^{N} \int_{\mathbb{T}^2} |\Delta\bphi_{n}|^{2}\,dV \right)^{1/4} + \lVert \nabla \bu \rVert_{L^2} \\
&\lesssim \frac{1}{4} \Gamtwo \sum_{n=1}^{N} \int_{\mathbb{T}^2} |\Delta\bphi_{n}|^{2}\,dV + \frac{3}{4}\Gamtwo^{-1/3} \lVert \nabla \bu \rVert_{L^2}^{4/3} \left(\sum_{n=1}^{N} \int_{\mathbb{T}^2} |\bphi_{n}|^{2}\,dV \right)^{1/3} + \lVert \nabla \bu \rVert_{L^2}\,,
\end{align*}
where we have used the Lieb-Thirring inequality \eqref{LT1A}, as well as the Cauchy-Schwartz inequality for sums. Using Young's inequality we now have
\begin{align*}
\Gamzero \sum_{n=1}^{N} \int_{\mathbb{T}^2} |\nabla\bphi_{n}|^{2}\,dV &\leq \frac{1}{4} \Gamtwo \sum_{n=1}^{N} \int_{\mathbb{T}^2} |\Delta\bphi_{n}|^{2}\,dV + \frac{\Gamzero^2}{\Gamtwo} \sum_{n=1}^{N} \int_{\mathbb{T}^2} |\bphi_{n}|^{2}\,dV \\
&= \frac{1}{4} \Gamtwo \sum_{n=1}^{N} \int_{\mathbb{T}^2} |\Delta\bphi_{n}|^{2}\,dV + \frac{\Gamzero^2}{\Gamtwo} N,
\end{align*}
where we have used property \eqref{trrelns2}. Inserting this bound into estimate \eqref{2dtraceestimate} for the volume elements gives
\begin{align*}
Tr \left[ \mathbf{P}_{N}\mathcal{M}\mathbf{P}_{N}\right] &\leq - \frac{1}{4} \Gamtwo \sum_{n=1}^{N} \int_{\mathbb{T}^2}|\Delta\bphi_{n}|^{2}dV + \Gamtwo^{-1/3} \lVert \nabla \bu \rVert_{L^2}^{4/3} N^{1/3}\\ 
&+ \left( \frac{\Gamzero^2}{\Gamtwo} - \alpha \right) N + \lVert \nabla \bu \rVert_{L^2} \\
&\leq - \frac{1}{4}\Gamtwo N^3 + \Gamtwo^{-1/3}\lVert \nabla \bu \rVert_{L^2}^{4/3} N^{1/3} + \left( \frac{\Gamzero^2}{\Gamtwo} - \alpha \right) N + \lVert \nabla \bu \rVert_{L^2} \,.
\end{align*}
Then, by taking time averages of this inequality and using inequality \eqref{LT6}, we obtain
\begin{align*}
\left<Tr \left[ \mathbf{P}_{N}\mathcal{M}\mathbf{P}_{N}\right] \right>_{T} &\leq - \frac{1}{2}\Gamtwo N^3 + \Gamtwo^{-1/3} \left\langle \lVert \nabla \bu \rVert_{L^2}^{4/3} \right\rangle_T N^{1/3} + \left( \frac{\Gamzero^2}{\Gamtwo} - \alpha \right) N \\
&+ \left\langle \lVert \nabla \bu \rVert_{L^2} \right\rangle_{T}.
\end{align*}
By using an interpolation inequality we have
\begin{align}\label{GNI1}
\left\langle \lVert \nabla \bu \rVert_{L^2}^{4/3} \right\rangle_T &\leq \left\langle \lVert \bu \rVert_{L^4}^{2/3} \lVert \Delta \bu \rVert_{L^2}^{2/3} \right\rangle_T \leq \left\langle \lVert \bu \rVert_{L^4}^4 \right\rangle_T^{1/6} \left\langle \lVert \Delta \bu \rVert_{L^2}^2 \right\rangle_T^{1/3}, \\
\left\langle \lVert \nabla \bu \rVert_{L^2} \right\rangle_{T} &\leq \left\langle \lVert \bu \rVert_{L^4}^{1/2} \lVert \Delta \bu \rVert_{L^2}^{1/2} \right\rangle_T \leq \left\langle \lVert \bu \rVert_{L^4}^4 \right\rangle_T^{1/8} \left\langle \lVert \Delta \bu \rVert_{L^2}^2 \right\rangle_T^{1/4}\,.
\end{align}
From estimate \eqref{L2absorbingball1} in the proof of Theorem \ref{weaksolutionsthm} we recall that
\begin{equation*}
\Rb \left\langle \lVert \bu \rVert_{L^4}^4 \right\rangle_T + \Gamtwo \left\langle \lVert \Delta \bu \rVert_{L^2}^2 \right\rangle_T \leq 8 \frac{|\alpha_{R}|^2}{\Rb}\,.
\end{equation*}
This then yields 
\begin{align*}
\left\langle \lVert \nabla \bu \rVert_{L^2}^{4/3} \right\rangle_T &\leq 8 \left(\frac{|\alpha_{R}|^2}{\Rb}\right)^{1/2}\Gamtwo^{-1/3}\Rb^{-1/6}\,; \\
\left\langle \lVert \nabla \bu \rVert_{L^2} \right\rangle_T &\leq 8 \left(\frac{|\alpha_{R}|^2}{\Rb}\right)^{3/8} \Gamtwo^{-1/4} \Rb^{-1/8}
\end{align*}
and so 
\begin{align*}
\left<Tr \left[ \mathbf{P}_{N}\mathcal{M}\mathbf{P}_{N}\right] \right>_{T} &\leq - \frac{1}{2}\Gamtwo N^3 + 
N^{1/3}\Gamtwo^{-1/3}\left(\frac{|\alpha_{R}|^2}{\Rb}\right)^{1/2}\Gamtwo^{-1/3}\Rb^{-1/6}\\
&+ \left( \frac{\Gamzero^2}{\Gamtwo} - \alpha \right) N + \left(\frac{|\alpha_{R}|^2}{\Rb}\right)^{3/8} \Gamtwo^{-1/4} \Rb^{-1/8} \,.
\end{align*}
By applying Proposition 2.1 and Theorem 3.3 in \cite[Chapter 5]{temam1997}, we obtain the following bounds on the dimension:
\begin{align}
N &\leq c \left(\frac{\Gamzero^2}{\Gamtwo^2} - \alpha \Gamtwo^{-1}\right)^{1/2}\,; \\
N &\leq c \left(\frac{|\alpha_{R}|^2}{\Rb}\right)^{3/16}\Gamtwo^{-5/8}\Rb^{-1/16}\,; \\
N &\leq c \left(\frac{|\alpha_{R}|^2}{\Rb}\right)^{1/8} \Gamtwo^{-5/12} \Rb^{-1/24}\,.
\end{align}
Notice that the second inequality implies the validity of the third inequality (in the standard parameter range).
\end{proof}

\subsection{Attractor dimension estimate in the three-dimensional case}\label{3datt}

\begin{theorem}\label{3dattthm}
An estimate of the Lyapunov dimension of the attractor $\mathcal{A}$ of the 3$d$ TTSH equations on the unit domain $\mathbb{T}^3$ is given by
\begin{equation*} \label{3dest}
d_L (\Att) \leq c \max \left\{\left\lvert\frac{\Gamzero^2}{\Gamtwo^{2}} - \alpha \Gamtwo^{-1} \right\rvert^{3/4}\,;~~
2^{9/4}\Gamtwo^{-1}\left(\frac{|\alpha_{R}|}{\Rb}\right)^{1/2}\right\}\,,
\end{equation*}
where we recall that $\alpha_{R} = \alpha - \frac{1}{2} \frac{\Gamzero^2}{\Gamtwo}$.
\end{theorem}
\par\vspace{-5mm}\noindent
\begin{remark} \label{3Dremark}
As observed in Theorem \ref{2dattthm} and also Remark \ref{2Dremark}, by keeping $\alpha$ and $\Rb$ fixed, and taking $\Gamtwo \ll 1$ and $\Gamzero \gg 1$ (i.e., the common parameter scaling in active turbulence), to leading order we therefore find the dominant term in the attractor dimension estimate to be (if $\Rb^{-1} \ll \Gamzero$)
\begin{equation}
d_L (\Att) \sim \left(\frac{\Gamzero^2}{\Gamtwo^{2}} -\alpha\Gamtwo^{-1}\right)^{3/4}\,,
\end{equation}
from which we recover the inverse Swift-Hohenberg scale $\eta_{res}^{-1} = (\Gamzero/\Gamtwo)^{1/2}$ as in the heuristic relation \eqref{Ndef}. A re-working of this calculation on the domain $[0,\,L]^3$ turns this estimate into $\eta_{res}^{-1}\sim L(\Gamzero/\Gamtwo)^{1/2}$.
\end{remark}
\par\smallskip\noindent
\begin{proof}[Proof of Theorem \ref{3dattthm}] The linearised three-dimensional TTSH equations, around a solution $\bu$, are given by
\beq{NS2}
\left(\partial_{t} + \bu\cdot\nabla\right)\delta\bu + \delta\bu\cdot\nabla\bu + \nabla \delta p 
&= -\left\{\alpha + \Gamzero \Delta + \Gamtwo \Delta^{2}\right\}\delta\bu\nonumber\\
&- \Rb\,\lvert \bu \rvert^{2}\delta\bu - 2\Rb\,\bu(\delta \bu \cdot \bu ).
\eeq
For any orthonormal subset of the basis for the space $W$ $\left\{\bphi_{1}, \bphi_{2},\ldots,\bphi_{N}\right\}$, the trace formula we need to estimate explicitly comes from \eqref{NS2} and is 
\begin{align*}
Tr \left[\mathbf{P}_{N}\mathcal{M}\mathbf{P}_{N}\right]  &= \sum_{n=1}^{N}
\int_{\mathbb{T}^3} \bphi_{n}\cdot\left\{-\alpha\bphi_{n} - \bu\cdot\nabla\bphi_{n}-\bphi_{n}\cdot\nabla\bu 
- \nabla\tilde{p}\left(\bphi_{n}\right)\right\}\,dV\\
&-\Rb\sum_{n=1}^{N}\int_{\mathbb{T}^3} \bphi_{n}\cdot\left\{\lvert \bu \rvert^{2}\bphi_{n} + 2 \bu (\bphi_{n} \cdot \bu) \right\}\,dV\,\\
&-\sum_{n=1}^{N}\int_{\mathbb{T}^3} \bphi_{n}\cdot\left\{\Gamzero \Delta\bphi_{n} + \Gamtwo \Delta^{2}\bphi_{n}\right\}\,dV\,.
\end{align*}
Since the elements $\bphi_{n}$ are divergence-free, the contribution from the pressure and the $\bphi_{n}\cdot\bu\cdot\nabla\bphi_{n}$-term vanishes. In the remaining term, coming from the advection, we integrate by parts, and we apply the Cauchy-Schwarz inequality (for sums) to obtain
\begin{align*}
&\left\lvert \sum_{n=1}^{N}
\int_{\mathbb{T}^3} \bphi_{n}\cdot\left\{\bphi_{n}\cdot\nabla\bu 
\right\}\,dV \right\rvert \leq \sum_{n=1}^{N}
\int_{\mathbb{T}^3} \lvert \bphi_{n} \rvert \lvert \nabla \bphi_{n} \rvert \lvert \bu \rvert \,dV \\
&\leq \int_{\mathbb{T}^3} \lvert \bu \rvert \left( \sum_{n=1}^N \lvert \bphi_n \rvert^2 \right)^{1/2} \left( \sum_{n=1}^N \lvert \nabla \bphi_n \rvert^2 \right)^{1/2}  \,dV.
\end{align*} 
What remains is
\begin{align*}
Tr \left[ \mathbf{P}_{N}\mathcal{M}\mathbf{P}_{N}\right] &\leq - \Gamtwo \sum_{n=1}^{N} \int_{\mathbb{T}^3} |\Delta\bphi_{n}|^{2}\!
dV + \Gamzero \sum_{n=1}^{N} \int_{\mathbb{T}^3} |\nabla\bphi_{n}|^{2}\,dV \\
& -\alpha \sum_{n=1}^{N} \int_{\mathbb{T}^3}|\bphi_{n}|^{2} \!dV + \int_{\mathbb{T}^3} \lvert \bu \rvert \left( \sum_{n=1}^N \lvert \bphi_n \rvert^2 \right)^{1/2} \left( \sum_{n=1}^N \lvert \nabla \bphi_n \rvert^2 \right)^{1/2}  \,dV.
\end{align*}
In the advection term we use the inequality $\|X\|_{L^{2}}^{2}\leq \|X\|_{L^{1}}^{\frac{q-2}{q-1}}\|X\|_{L^{q}}^{\frac{q}{q-1}}$ (for $q>2$), with $q=7/3$, and the Lieb-Thirring inequality \eqref{LT1B}, to find 
\begin{align*}
&\int_{\mathbb{T}^3} \lvert \bu \rvert \left( \sum_{n=1}^N \lvert \bphi_n \rvert^2 \right)^{1/2} \left( \sum_{n=1}^N \lvert \nabla \bphi_n \rvert^2 \right)^{1/2}  \,dV \\
&\leq \lVert \bu \rVert_{L^4} \left\lVert \left( \sum_{n=1}^N \lvert \nabla \bphi_n \rvert^2 \right)^{1/2} \right\rVert_{L^2} \left\lVert \left( \sum_{n=1}^N \lvert \bphi_n \rvert^2 \right)^{1/2} \right\rVert_{L^4} \\
&\leq \lVert \bu \rVert_{L^4} \left\lVert \sum_{n=1}^N \lvert \nabla \bphi_n \rvert^2 \right\rVert_{L^1}^{1/2} \left\lVert \sum_{n=1}^N \lvert \bphi_n \rvert^2 \right\rVert_{L^2}^{1/2} \\
&\leq \lVert \bu \rVert_{L^4}\left(\sum_{n=1}^{N}\int_{\mathbb{T}^3} |\bphi_{n}|^{2}\,dV\right)^{1/4} \left(\sum_{n=1}^{N} \int_{\mathbb{T}^3} |\Delta\bphi_{n}|^{2}\,dV \right)^{1/4}\\
&\times\left\lVert \sum_{n=1}^N \lvert \bphi_n \rvert^2 \right\rVert_{L^1}^{1/16} \left\lVert \sum_{n=1}^N \lvert \bphi_n \rvert^2 \right\rVert_{L^{7/3}}^{7/16} \\
&\leq \lVert \bu \rVert_{L^4}N^{5/16} \bigg[ \left(\sum_{n=1}^{N} \int_{\mathbb{T}^3} |\Delta\bphi_{n}|^{2}\,dV \right)^{7/16} + c \bigg] \\
&\leq \frac{1}{2}\Gamtwo \sum_{n=1}^{N} \int_{\mathbb{T}^3} |\Delta\bphi_{n}|^{2}\,dV + 4\Gamtwo^{-7/9} N^{5/9} \lVert \bu \rVert_{L^4}^{16/9} + c \lVert \bu \rVert_{L^4}N^{5/16}\,,
\end{align*}
where we have used the Lieb-Thirring inequality \eqref{LT1B} and estimate \eqref{LT2}. Now we obtain by using property \eqref{trrelns2}
\beq{tr4}
Tr \left[ \mathbf{P}_{N}\mathcal{M}\mathbf{P}_{N}\right] &\leq& - \frac{1}{2} \Gamtwo \sum_{n=1}^{N} \int_{\mathbb{T}^3} |\Delta\bphi_{n}|^{2}dV + \Gamzero \sum_{n=1}^{N} \int_{\mathbb{T}^3} |\nabla\bphi_{n}|^{2}\,dV\nonumber\\
&-&  \alpha \sum_{n=1}^{N} \int_{\mathbb{T}^3}|\bphi_{n}|^{2}dV + 4\Gamtwo^{-7/9} N^{5/9} \lVert \bu \rVert_{L^4}^{16/9} + c \lVert \bu \rVert_{L^4}N^{5/16} \nonumber\\
&\leq& - \frac{1}{2} \Gamtwo \sum_{n=1}^{N} \int_{\mathbb{T}^3} |\Delta\bphi_{n}|^{2}dV + \Gamzero \sum_{n=1}^{N} \int_{\mathbb{T}^3} |\nabla\bphi_{n}|^{2}dV \nonumber\\
&+& 4\Gamtwo^{-7/9} N^{5/9} \lVert \bu \rVert_{L^4}^{16/9} - \alpha N + c \lVert \bu \rVert_{L^4}N^{5/16}\,.\nonumber
\eeq
Applying Young's inequality yields
\begin{align*}
\Gamzero \sum_{n=1}^{N} \int_{\mathbb{T}^3} |\nabla\bphi_{n}|^{2}\,dV &\leq \frac{1}{4} \Gamtwo \sum_{n=1}^{N} \int_{\mathbb{T}^3} |\Delta\bphi_{n}|^{2}\,dV + \frac{\Gamzero^2}{\Gamtwo} \sum_{n=1}^{N} \int_{\mathbb{T}^3} |\bphi_{n}|^{2}\,dV \\
&= \frac{1}{4} \Gamtwo \sum_{n=1}^{N} \int_{\mathbb{T}^3} |\Delta\bphi_{n}|^{2}\,dV + \frac{\Gamzero^2}{\Gamtwo} N\,.
\end{align*}
Then we have
\begin{align*}
Tr \left[ \mathbf{P}_{N}\mathcal{M}\mathbf{P}_{N}\right] 
&\leq - \frac{1}{4} \Gamtwo \sum_{n=1}^{N} \int_{\mathbb{T}^3} |\Delta\bphi_{n}|^{2}dV\\
&+ 4 \Gamtwo^{-7/9} N^{5/9} \lVert \bu \rVert_{L^4}^{16/9} + N \left(-\alpha+ \frac{\Gamzero^2}{\Gamtwo} \right)+ c \lVert \bu \rVert_{L^4}N^{5/16} \\
&\leq - \frac{1}{4} \Gamtwo N^{7/3} + N \left(-\alpha+ \frac{\Gamzero^2}{\Gamtwo} \right) + 4\Gamtwo^{-7/9} N^{5/9} \lVert \bu \rVert_{L^4}^{16/9}\\ 
&+ c \lVert \bu \rVert_{L^4}N^{5/16}\,, 
\nonumber
\end{align*}
where we have used inequality \eqref{LT6}. From estimate \eqref{L2absorbingball1} we deduce that
\begin{align} 
&\Rb \left\langle \lVert \bu \rVert_{L^4}^4 \right\rangle_T \leq 8 \frac{|\alpha_{R}|^2}{\Rb}\,,
\end{align}
which then implies
\begin{align*}
\left\langle \lVert \bu \rVert_{L^4}^{16/9} \right\rangle_T &\leq \left\langle \lVert \bu \rVert_{L^4}^4 \right\rangle_{T}^{4/9} \leq \left(4 \frac{|\alpha_{R}|}{\Rb}\right)^{8/9}\,, \\
\left\langle \lVert \bu \rVert_{L^4} \right\rangle_{T} & \leq \left\langle \lVert \bu \rVert_{L^4}^4 \right\rangle_{T}^{1/4} \leq \left(4 \frac{|\alpha_{R}|}{\Rb}\right)^{1/2}\,.
\end{align*}
By taking time averages we find
\begin{align*}
&\left<Tr \left[ \mathbf{P}_{N}\mathcal{M}\mathbf{P}_{N}\right] \right>_{T} \leq - \frac{1}{4} \Gamtwo N^{7/3} + N \left(-\alpha+ \frac{\Gamzero^2}{\Gamtwo} \right) + 4N^{5/9} \Gamtwo^{-7/9} \left\langle \lVert \bu \rVert_{L^4}^{16/9} \right\rangle_T \\
&+ c \left\langle \lVert \bu \rVert_{L^4} \right\rangle_T N^{5/16} \\
&\leq - \frac{1}{4}\Gamtwo N^{7/3} + N \left(-\alpha+ \frac{\Gamzero^2}{\Gamtwo} \right) + 
4 N^{5/9} \Gamtwo^{-7/9}\left(4 \frac{|\alpha_{R}|}{\Rb}\right)^{8/9}\\ &+ c \left(4 \frac{|\alpha_{R}|}{\Rb}\right)^{1/2} N^{5/16}\,.
\end{align*}
We can estimate the value $N_{0}$ for which
$\left<Tr \left[ \mathbf{P}_{N}\mathcal{M}\mathbf{P}_{N}\right] \right>_{T} < 0$ when $N>N_{0}$.  This leads to the following inequalities:
\begin{align*}
N_0^{4/3} &\leq c \Gamtwo^{-1} \left(-\alpha + \frac{\Gamzero^2}{\Gamtwo} \right) =  \left(\frac{\Gamzero^2}{\Gamtwo^2}
-\alpha \Gamtwo^{-1}\right)\,; \\
N_0^{16/9} &\leq c 16\Gamtwo^{-16/9} \left(\frac{|\alpha_{R}|}{\Rb}\right)^{8/9}\,, \\
N_0^{97/48} & \leq c \Gamtwo^{-1} \left(4 \frac{|\alpha_{R}|}{\Rb}\right)^{1/2}\,. 
\end{align*}
We can observe that the third inequality is implied by the second one in the standard value range of the control parameters. Thus, we have 
\begin{equation}
N_{0} \sim \max \left\{\left\lvert- \alpha \Gamtwo^{-1} + \frac{\Gamzero^2}{\Gamtwo^2} \right\rvert^{3/4}\,,~~
2^{9/4}\Gamtwo^{-1}\left(\frac{|\alpha_{R}|}{\Rb}\right)^{1/2}\right\}\,.
\end{equation}
The result then follows from Proposition 2.1 and Theorem 3.3 in \cite[Chapter 5]{temam1997}.
\end{proof}

\section{\large Numerical evaluation of Lyapunov exponents and $d_L$}\label{Kiran}

We now present numerical estimates for the Lyapunov spectrum, $\mu_{n}$ and the Lyapunov dimension in Eq.~\eqref{dldef} for the $d=2$ case. In Theorem \ref{2dattthm}, based on a unit domain $[0,\,1]^2$, we estimated the number of exponents $N_{0}$ to be $N_{0} \sim \Gamzero/\Gamtwo$ (at leading order). As already noted in Remark \ref{2Dremark}, the following numerical calculation requires an adjustment of the domain length $L$. A re-working of this calculation on the domain $[0,\,L]^2$ turns this estimate into $N_{0} \sim L^{2}\Gamzero/\Gamtwo$.
\par\smallskip
The number of Lyapunov exponents is equal to the total number of degrees of freedom available numerically. Computing the Lyapunov spectrum for the full set of hydrodynamic equations, and consequently $d_{L}$, is computationally expensive~\cite{KMK1992,HR2019JCP}. However, as we describe below, it is not necessary to extract the full Lyapunov spectrum to compute $d_{L}$; we need only $N_{0}$ exponents, which is smaller than the total number of degrees of freedom. Theorem \ref{globalattractorthm} shows that $\Att_{L^2} = \Att_{H^1}$, so it makes no difference whether we work with velocity or vorticity fields. In fact, in \S\ref{2datt} our estimate for $d_{L}$ has been made using the velocity field; but to reduce numerical costs, we simulate the attractor dimension using the equations corresponding to the vorticity field $\bom= \mbox{curl}\,\bu$.
\begin{equation}{\label{s1a}}
\left(\partial_{t}+\bu\cdot\nabla\right)\bom = -\left(\alpha+\Gamzero\Delta + \Gamtwo\Delta^{2}\right)\bom - \Rb\mbox{curl}\,(\bu|\bu|^{2})\,. 
\end{equation}
\begin{figure}[!]
\includegraphics[width=\linewidth]{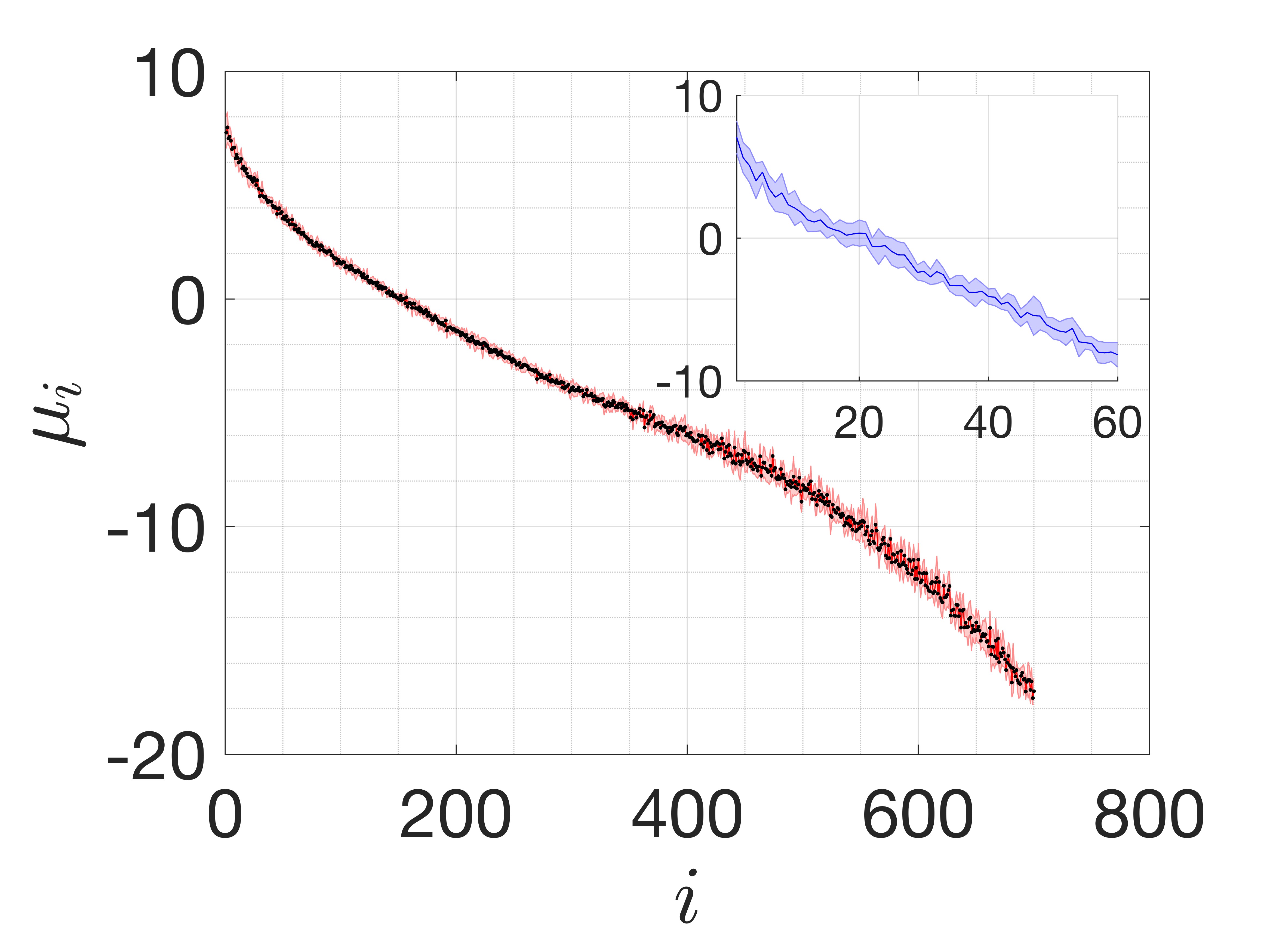}
\caption{The Lyapunov spectrum for the run 7, where the black dots are for the unordered spectrum and the red curve gives the ordered spectrum, with shaded region giving the error bars. The inset gives the ordered Lyapunov spectrum for the run 2.}\label{fig:lyp_conv}
\end{figure}
We begin by choosing a base vorticity field, $\bom_{0}$, to which we add a set of $i=1,..M$ initial perturbations $\{\delta\bom_{i}\}$ to get $M$ initial conditions (which we evolve indepedently)\,:
\bel{kvk_pert1}
\bom_{i}=\bom_{0}+\epsilon\delta\bom_{i}, 
\ee
where we use $\epsilon= 10^{-6}$. The initial perturbations are chosen to be orthonormal, i.e.,
\bel{kvk_on}
\left(\delta \bom_{i},\delta \bom_{j}\right)_{L^{2}}=\delta_{ij}\,,
\ee
where $\delta_{ij}$ is the Kronecker delta and $(.)_{L^{2}}$ denotes the standard $L^{2}$ inner product, e.g., for any real functions $f$ and $g$ defined on $V_{d}$:
\bel{l2norm}
(f,g)_{L^{2}}=\int_{V_{d}} fg~d^{d}\mathbf{x}.
\ee
\begin{figure}[!]
\includegraphics[width=\linewidth]{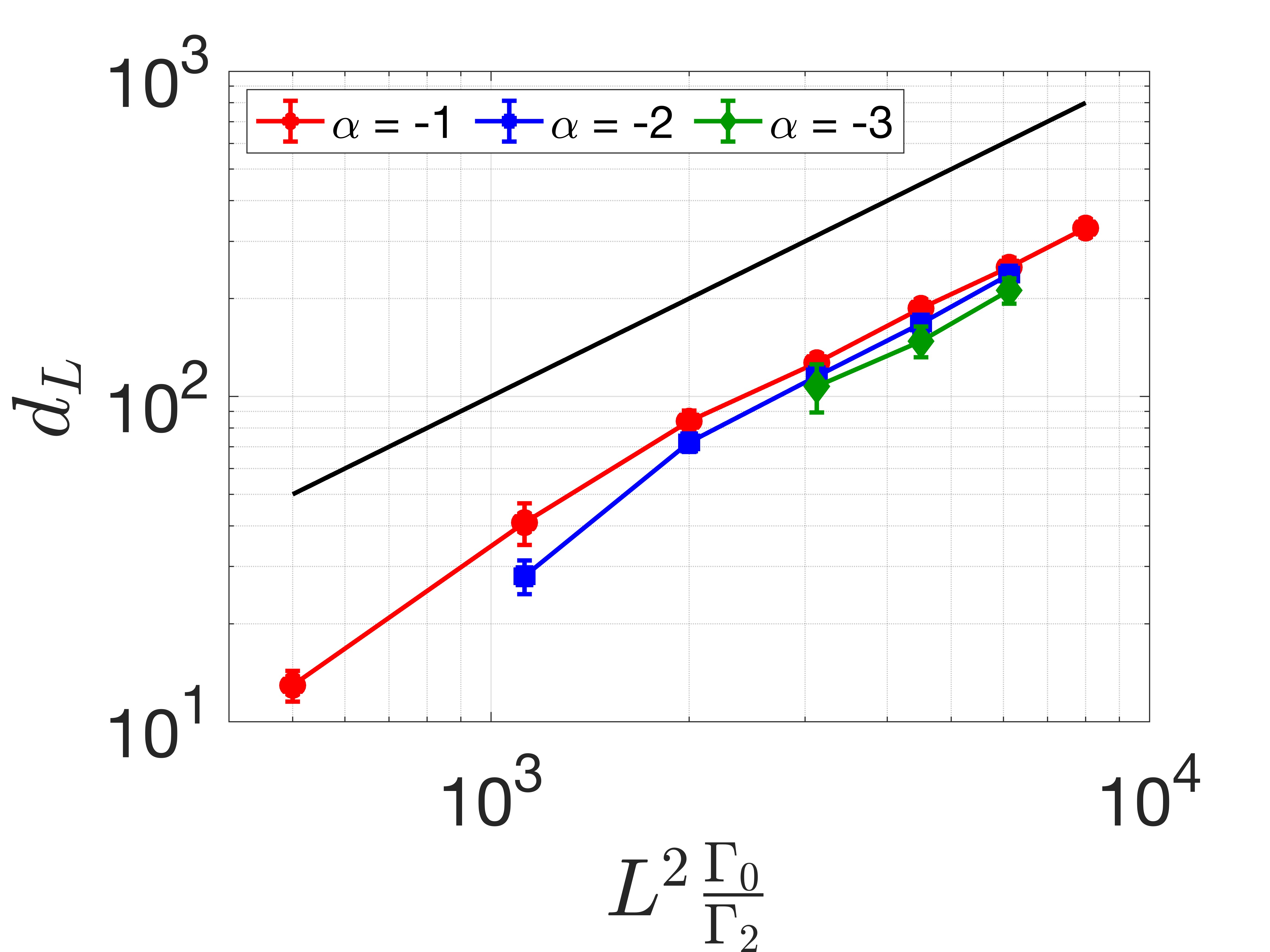}

    \caption{Log-log plot of $d_{L}$ versus $L^2(\Gamzero/\Gamtwo)$ for different values of $\alpha$. The solid black line corresponds to  $10^{-1}\times L^2(\Gamzero/\Gamtwo)$.}
    \label{fig:attdim}
\end{figure}
We then define a norm using the above inner-product, namely, $\vert\vert f\vert\vert$ $\coloneqq$ $\sqrt{(f,f)_{L^{2}}}$. Choosing the base vorticity field as a steady-state solution of Eq.~\eqref{s1a}, the set of $M$ initial conditions $\{\bom_{i}\}$ and $\bom_{0}$ are then evolved simultaneously using the full non-linear equations Eq.~\eqref{s1a}. We estimate the Lyapunov exponents by calculating the finite-time Lyapunov exponents (FTLE) $\{\gamma_{i}\}$:
\bel{kvk_FTLE}
\gamma_{i}(t,t_{0})=\frac{1}{t-t_{0}}\log\frac{\vert\vert\delta\bom_{i}(t)\vert\vert}{\vert\vert\delta\bom_{i}(t_{0})\vert\vert},
\ee
The Lyapunov exponents are defined as the asymptotic limits to FTLEs~\cite{JM2015}:
 \bel{kvk_LE1}
\mu_{i}=\lim_{t\to\infty}\gamma_{i}(t,t_{0}).
 \ee
 With the assumptions of ergodicity~\cite{Ose1968,JM2015,HR2019JCP,HR2019PRF}, the Lyapunov exponents are simply the ensemble averages of the FTLEs:
 \bel{kvk_LE2}
 \mu_{i}=\langle\gamma_{i}\rangle =\frac{1}{t-t_{0}}\bigg\langle\log\frac{\vert\vert\delta\bom_{i}(t)\vert\vert}{\vert\vert\delta\bom_{i}(t_{0})\vert\vert}\bigg\rangle,
 \ee
where $\langle \cdot \rangle$ denotes the ensemble average (see below). 
\par\smallskip
When estimating $\{\gamma_{i}\}$  and the corresponding Lyapunov exponents numerically, there are two important challenges. Firstly, the equations do not preserve orthonormality of the perturbations and all the perturbations collapse onto the fastest-growing direction giving just the largest Lyapunov exponent. Secondly, since we are using the full nonlinear equations, the growth $\delta\bom_{i}$ is exponential only for small times, which corresponds to small $\vert\vert\delta\bom_{i}\vert\vert$. We resolve both these issues by periodically re-orthogonalizing the perturbations and rescaling their norms by $\epsilon$. At the time of re-orthogonalization we measure $\{\gamma_{i}\}$ using Eq.~\eqref{kvk_FTLE}. We repeat this process multiple times resulting in statistically independent sets for $\{\gamma_{i}\}$ that constitute the ensemble \cite{GL1987PRL,GL1991,KMK1992,HR2019JCP,HR2019PRF,BC2019PRE,CTB2020PRF}. We use a pseudo spectral method to solve Eq.~\eqref{s1a}. For the time integration in Fourier-space we use a second-order integrating factor Runge-Kutta method~\cite{KGVP2023}. Finally, we use a modified Gram-Schmidt algorithm to orthogonalize the perturbations in the real space and find that it is sufficient to orthogonalize every 50 iterations.
\par\vspace{3mm}\noindent
To estimate $d_{L}$, we need enough positive and negative Lyapunov exponents such that Eq.~\eqref{KY2} is satisfied. In our simulations, we have $M$ in the range between $50$ to $1000$. After large ensemble averages, the Lyapunov exponents are ordered as follows: $\mu_{1}>\mu_{2}>\ldots>\mu_{M}$\,; this allows us to check if the resulting spectrum has converged and to obtain enough Lyapunov exponents to compute $d_{L}$~\cite{BC2019PRE}. In Fig.~\ref{fig:lyp_conv} we plot the Lyapunov spectra for the two representative runs 7 and 2 (the latter in the inset). The black dots indicate the unordered Lyapunov exponents and the red curve corresponds to the same exponents, but ordered. We find reasonable agreement between the spectra for ordered and unordered exponents; this indicates good convergence~\cite{BC2019PRE}. 

A few comments are in order: we denote by $\sigma_{i}$ the standard deviation of the $i^{th}$ exponent $\mu_{i}$; for small $i$, $\mu_{i}/\sigma_{i}\sim \mathcal{O}(1)$, i.e,, the largest Lyapunov exponents have large errors. In contrast, for very large $i$, i.e., very negative Lyapunov exponents, $\mu_{i}/\sigma_{i} \sim \mathcal{O}(10^{-2})$, so the errors are small. These observations are consistent with Lyapunov-spectrum measurements in two and three dimensions for homogeneous isotropic and incompressible flows ~\cite{HR2019PRF,BC2019PRE,CTB2020PRF}.  
\begin{table}[!h]
		\centering
		\begin{tabular}{|c|c| c|}
			\hline
			Run & $L$ &  $\alpha$ 
			\\
			\hline
			1& 1& -1\\
			\hline
			2& 1.5& -1\\
			\hline
			3& 2& -1\\
            \hline
			4& 2.5& -1\\
			\hline
			5& 3& -1\\
            \hline
			6& 3.5& -1\\
            \hline
			7& 4.0& -1\\
			\hline
			8& 3.5& -2\\
			\hline
			9& 3.5& -3\\
			\hline
			10& 3& -3\\
			\hline
			11& 3& -2\\
			\hline
			12& 2.5& -3\\
			\hline
			13& 2.5& -2\\
			\hline
			14& 1.5& -2\\
			\hline
			15& 2.0& -2\\
			\hline
		\end{tabular}
        \medskip
        \caption{\footnotesize List of the parameters for our DNSs\,: $L$ is the length of the periodic square. We fixed the parameters $\beta=0.5$, $\Gamma_0=-0.045$, $\Gamma_2=9\times10^{-5}$, $\lambda=3.5$, and $N=128$ in all our simulations.
        }\label{tab:table_DNS}
\end{table}
\par\vspace{0mm}\noindent
From Theorem 5, the leading-order term for the upper-bound of the attractor-dimension scales as $d_L\sim L^2(\Gamzero/\Gamtwo)$. In Fig.~\ref{fig:attdim} we plot $d_{L}$ versus $L^2(\Gamzero/\Gamtwo)$ for different values of $\alpha$. We note that the numerically estimated values scale as $L^2(\Gamzero/\Gamtwo)$ (solid black line) for $L^2(\Gamzero/\Gamtwo)>10^3$; the dependence of these plots on $\alpha$ is marginal. 

\section{\large Conclusions}\label{conclusions}

The TTSH equations in two and three space dimensions are used to model dense bacterial suspensions~\cite{AJC2020,ACJ2022,Pandit_2025}, and other active turbulent systems. We have studied these equations from the point of view of rigorous analysis. In particular, in the case of periodic boundary conditions, we have proved that they possess a global attractor whose Lyapunov dimension is bounded. By interpreting this as the number of degrees of freedom in the system, we have shown that the estimates are dominated by the linear terms, which means that we have established the predominance of the Swift-Hohenberg scale as the natural resolution scale.
\par\smallskip
We have also carried out pseudospectral direct numerical simulations of these PDEs in two dimensions, obtained Lyapunov spectra and the Kaplan-Yorke dimensions for representative parameter values, and we have shown that our numerical results are consistent with the rigorous bounds that have been derived analytically.

\section{\large Acknowledgements}\label{acknowl}

K.V. Kiran and R. Pandit thank the Anusandhan National Research Foundation (ANRF), the Science and Engineering Research Board (SERB), and the National Supercomputing Mission (NSM), India, for support, and the Supercomputer Education and Research Centre (IISc), for computational resources. D.W. Boutros would like to acknowledge support from the Cambridge Trust and the Cantab Capital Institute for Mathematics of Information. K.V. Kiran acknowledges support from the Simons' Collaboration on Wave Turbulence (award No. 651459).
The authors would like to thank the Isaac Newton Institute for Mathematical Sciences, Cambridge, for support and hospitality during the programme \textit{Anti-diffusive dynamics\,: from sub-cellular to astrophysical scales}, where some of the work on this paper was undertaken. This work was supported by EPSRC grant no EP/R014604/1. 

\bibliographystyle{abbrv}
\bibliography{ttsh}

\end{document}